\documentclass[11pt,a4paper]{article}

\usepackage[margin=1in]{geometry}
\usepackage[utf8]{inputenc}
\usepackage[T1]{fontenc}

\usepackage{amsthm}
\usepackage{amsmath}
\usepackage{amssymb}
\usepackage{amsfonts}
\usepackage{graphicx}
\usepackage{hyperref}
\usepackage{cite}
\usepackage{microtype}
\usepackage{booktabs}
\usepackage{lmodern}
\usepackage{url}
\usepackage{authblk}
\usepackage{color}

\usepackage{fancyvrb}
\usepackage{fvextra} 

\usepackage{pgfplots}
\pgfplotsset{compat=1.17}

\usetikzlibrary{arrows.meta, positioning, decorations.pathreplacing, calc}

\newtheorem{definition}{Definition}
\newtheorem{lemma}{Lemma}
\newtheorem{problem}{Problem}
\newtheorem{example}{Example}
\newtheorem{theorem}{Theorem}
\newtheorem{remark}{Remark}

\newtheorem{proposition}[theorem]{Proposition} 

\long\def\comment#1{}

\newcommand{\UP}[0]{P_\upsilon}
\newcommand{\UNP}[0]{NP_\upsilon}
\newcommand{\UBQP}[0]{BQP_\upsilon}

\newtheorem{corollary}[theorem]{Corollary}

\begin{document}

\title{On the Computational Complexity of Problems: Formalizing Sensitivity to Uncertainty and Parametric Complexity Classes}

\author[1,2,3]{Yannis Tzitzikas}

\affil[1]{Department of Computer Science, University of Crete, Greece}
\affil[2]{Institute of Computer Science, FORTH, Heraklion, Greece}
\affil[3]{Center for Quantum Science \& Technologies,  FORTH-QuTech, Heraklion, Greece}


\maketitle

\begin{abstract}
Worst-case analysis, the setting of the $P \stackrel{?}{=} NP$ question, measures the cost of a problem as a function of input size alone. Yet problems that are exponential in general often become polynomial, even linear, once a checkable property of the input is known: Partition, for instance, is decidable in linear time on inputs whose elements are all equal. We develop a framework that makes such knowledge an explicit parameter $K$ of complexity statements. Costs are measured for fixed algorithms; we distinguish the promise cost from the guarded cost that also pays for checking $K$, and introduce the residual uncertainty $\Psi(n,K)$, the number of input cells that, together with $K$, determine the input. An Uncertainty Reduction Theorem decomposes the guarded cost into checking $K$, locating and reading the residual information, and solving an explicit reduced problem. The classes $\UP[K]$ and $\UNP[K]$ are restrictions to polynomial-time promises, made informative by the coverage of $K$. We show that small $\Psi$ forces sparsity but not tractability, characterize FPT as uniform membership in graded families $\UP[K_k]$, prove a free-variable threshold for SAT, and show that under the factoring assumption no checkable property of non-negligible coverage makes factoring easy. We relate the framework to promise problems, kernelization, average- and generic-case complexity and beyond-worst-case analysis, illustrate it on Partition, Graph Isomorphism, SAT, Vertex Cover and TSP, and apply it to guarded algorithm dispatch with latency contracts and optimal guard ordering, to solver routing in agentic AI, and to the role of promises in quantum speedups.
\end{abstract}


\section{Introduction}

The central question in theoretical computer science is the $P$ vs. $NP$ problem, which asks whether every problem whose solution can be quickly verified can also be quickly solved \cite{Fortnow2009TheSO}. While this question is of immense theoretical importance and represents one of the seven Millennium Prize Problems, its classical formulation possesses distinct practical limitations. In the real world, the fact that a problem belongs to the $NP$-complete class often discourages practitioners and software engineers from developing applications that rely on solving it. However, this pessimistic worst-case characterization is frequently a poor representation of practical difficulty. Many computationally hard tasks exhibit highly tractable behavior in practice because real-world problem instances are rarely random, unstructured, or worst-case; instead, they satisfy specific structural properties, symmetries, or constraints that drastically reduce computational complexity.

To address this gap, this paper does not question the importance of the $P$ vs.\ $NP$ question, which by definition concerns the unrestricted problem. Instead, it asks a complementary question that practitioners face daily: \emph{given that the input is known to satisfy a checkable property $K$, how expensive is the problem, and what does it cost to exploit $K$?} By focusing on the sensitivity of computational tasks to {\em input uncertainty}, we make such knowledge an explicit parameter of complexity statements and study the induced parametric complexity classes $\UP[K]$ and $\UNP[K]$, where $K$ denotes a property. When extra knowledge (which lessens uncertainty) is available and exploited, the cost of solving a hard problem can drop from exponential to polynomial or even linear time (a constant-time core step following verification). Nothing in this changes the status of $P$ vs.\ $NP$; what changes is that, for a given application, the relevant question becomes whether its instances satisfy a property $K$ that is cheap to check and that makes the problem easy.

Much of this intuition is classical, and the framework should be read against several established lines of work (\S\ref{sec:RelatedWork}): promise problems, parameterized complexity and kernelization, average-case and generic-case complexity, instance optimality and beyond-worst-case analysis, and algorithm selection. Our aim is a single accounting in which the property $K$, the cost of verifying it, the cost of locating the information that remains unknown, and the size of that residual information all appear explicitly. One might ask, in particular, how this relates to {\em average-case complexity} \cite{levin1986avg,bogdanov2006avg}, which measures the expected cost over a distribution of inputs. Average-case complexity asks whether the \emph{typical} instance is easy; the present framework asks whether there is a \emph{checkable} structural property $K$ such that \emph{every} instance satisfying $K$ is easy. The former is about typicality, the latter about {\em certifiability}. If a problem $P$ is polynomial on inputs satisfying $K$ (i.e.\ $P\in\UP[K]$), then it is average-case polynomial for every distribution concentrated on $\{x\models K\}$; conversely, errorless average-case algorithms, which never err but may answer ``don't know'', implicitly define such a set, but not necessarily one that is cheaper to check than running the algorithm. The precise relation is analyzed in \S\ref{sec:Originality}.

{\bf Contributions.} (1) A cost model based on fixed algorithms (\S\ref{sec:SU}), which avoids the degeneracy of per-instance infima, with promise and guarded versions of the restricted worst-case cost (\S\ref{sec:MeasuringSU}). (2) The residual-uncertainty metric $\Psi(n,K)$, its task-relative and accessible variants, and its relation to counting and to Kolmogorov complexity (\S\ref{sec:MeasuringSU}). (3) An Uncertainty Reduction Theorem that bounds the guarded cost through an explicitly defined reduced problem, and its relation to kernelization (\S\ref{sec:UncertaintyReductionTheorem}). (4) Structural results on the classes $\UP[K]$ and $\UNP[K]$: their identification with polynomial-time promises, the role of coverage, the consequences of small $\Psi$, the characterization of FPT as uniform membership in a graded family of properties, a free-variable threshold for SAT, and a coverage-based notion of knowledge resistance, instantiated for factoring (\S\ref{sec:Parametric}--\S\ref{sec:NPProblemsSensitiveInsensitive}). (5) Case studies and applications to guarded dispatch with latency contracts, guard ordering, agentic AI and quantum computing (\S\ref{sec:Originality}--\S\ref{sec:SE}).

\subsection{Related Work}
\label{sec:RelatedWork}
\emph{Promise problems} \cite{evenselman1984promise,goldreich2006promise} restrict a decision problem to inputs satisfying a predicate; when the predicate is decidable in polynomial time, the classes $\UP[K]$ and $\UNP[K]$ studied here are exactly restrictions of this kind (\S\ref{sec:Parametric}). \emph{Parameterized complexity} \cite{downey1999,cygan2015parameterized} measures cost as a function of the input size and a parameter $k$; \emph{kernelization} \cite{fomin2019kernelization}, i.e.\ polynomial-time reduction to an equivalent instance whose size depends only on $k$, is the closest classical analogue of the reduction to the $\Psi(n,K)$ determining cells used in our main theorem. Structural restrictions that make NP-hard problems polynomial (graph classes, bounded treewidth, dichotomy theorems such as Schaefer's \cite{schaefer1978complexity}) supply most of the concrete properties $K$ used in this paper. \emph{Average-case complexity} \cite{levin1986avg,ben-david1992avg,bogdanov2006avg,impagliazzo1995personal} and \emph{generic-case complexity} \cite{kapovich2003generic} ask whether typical, respectively almost all, instances are easy. \emph{Instance optimality} \cite{fagin2003optimal,afshani2017instance}, \emph{smoothed analysis} \cite{spielman2004smoothed} and the broader program of \emph{beyond-worst-case analysis} \cite{roughgarden2021beyond} refine worst-case bounds using instance features or perturbations; our cost model is deliberately simpler (fixed algorithms, worst case over $\{x\models K\}$). Finally, routing an instance to a specialized solver is Rice's \emph{algorithm selection problem} \cite{rice1976algorithm}, realized in portfolio solvers such as SATzilla \cite{xu2008satzilla}, and ordering guards by a cost/probability ratio is an instance of Smith's rule \cite{smith1956various} and of optimal sequential testing \cite{simon1975optimal}. The contribution of this paper is to combine these ingredients into one explicit accounting, not to supersede any of them.

The remainder of this paper is organized as follows.
Section~\ref{sec:TheRoleOfUncertainty} motivates the role of uncertainty through a simple game related to the Two Envelopes Paradox.
Section~\ref{sec:SU} fixes the cost model, introduces Sensitivity to Uncertainty, and presents a case study on the Partition Problem.
Section~\ref{sec:MeasuringSU} defines restricted, promise and guarded costs and the residual-uncertainty metric $\Psi(n,K)$, and relates $\Psi$ to counting and to Kolmogorov complexity.
Section~\ref{sec:UncertaintyReductionTheorem} states and proves the Uncertainty Reduction Theorem.
Section~\ref{sec:Parametric} introduces the classes $\UP[K]$ and $\UNP[K]$, coverage and the Uncertainty Sensitivity Index (USI), and analyzes what small residual uncertainty implies.
Section~\ref{sec:NPProblemsSensitiveInsensitive} studies Graph Isomorphism, the knowledge resistance of factoring, the relation to FPT, and a free-variable threshold for SAT.
Section~\ref{sec:Originality} positions the framework and presents further case studies (Partition, SAT, Vertex Cover and TSP).
Section~\ref{sec:SE} discusses applications in software engineering, agentic AI and quantum computing.
Finally, Section~\ref{sec:CR} concludes the paper.


\section{The Role of Uncertainty through a Simple Game}
\label{sec:TheRoleOfUncertainty}

To understand the fundamental role of uncertainty and unknownness in computational complexity, consider a simple game played between two individuals, 
say Yannis and Tonia. Yannis shows his closed fists to Tonia and tells her:
\begin{quote}
\textit{``Tonia, in each of my hands I have a monetary coin. The value of one hand is double the value of the other. Which hand do you select?''}
\end{quote}

Tonia selects Yannis's left hand. Yannis opens his palm and Tonia sees that there is exactly $1$ Euro in that hand. Yannis then poses the following subsequent question:
\begin{quote}
\textit{``Would you like to take this Euro, or would you like to take whatever my right hand has?''}
\end{quote}

Tonia performs an initial probabilistic reflection: if the left hand contains $1$ Euro (let this be $x = 1$), the right hand must contain either $0.5$ Euro ($x/2$) or $2$ Euros ($2x$). Assuming an equal probability of $0.5$ for each scenario, she computes her expected gain if she switches to the right hand, as
$
\text{Expected Gain} = \frac{2 + 0.5}{2} = 1.25\text{ Euros}.
$
This expected value is strictly greater than $1$ Euro, suggesting that she should change hands. However, on second thought, Tonia realizes this reasoning is fallacious: there is no rational reason to select one hand and then switch to the other in a symmetric state, as she could repeat this switching reasoning infinitely. This is a manifestation of the classic \textit{Two Envelopes Paradox} \cite{broome1995two}.

Her initial reflection was wrong because it mixed up the notion of structural uncertainty or unknownness with probability. The values in Yannis's hands are already fixed and certain; they are not random variables drawn from a known probability distribution. The actual sample space is either $\{0.5\text{ Euro}, 1\text{ Euro}\}$ or $\{1\text{ Euro}, 2\text{ Euros}\}$. The uncertainty lies in which sample space Yannis chose, which is a structural unknown, not a probability distribution. In the standard analysis of the paradox the same point appears in probabilistic form: the step that fails is the assumption that, whatever amount is observed, the other hand is equally likely to hold half or double that amount, since no proper prior distribution over the amounts has this property for every amount \cite{broome1995two}. Either way, the probabilities Tonia used were not justified by anything she knew.

The analogy we want to draw is limited but, we believe, useful. The $P \stackrel{?}{=} NP$ question is, by definition, a question about \emph{all} inputs of a given size, and as such it is perfectly well posed. The notions of best, average, and worst-case complexity, however, all encode a stance towards our uncertainty about the input, and the worst case encodes the most pessimistic one: it treats the input as chosen by an adversary from the whole of $I_n$. In applications we usually know more, not in the form of a probability distribution (which, as in the game, we often cannot justify), but in the form of \emph{facts} about the instance that can be checked. The rest of the paper makes such facts an explicit parameter $K$ of complexity statements, and charges the cost of checking and exploiting them.

\section{Computational Tasks and Sensitivity to Uncertainty}
\label{sec:SU}

We shall now formalize how the cost of solving a problem {\em varies} across instances of the same size. Informally, \textit{Sensitivity to Uncertainty (SU)} measures how much the computational cost of carrying out a task varies across different instances of the same size.

\paragraph{Computational model.} We work in the unit-cost word RAM: the input is an array $A[1..n]$ of $n$ \emph{cells}, each holding a word over a cell alphabet $\Sigma$ (e.g.\ $\Sigma=\{0,1\}^{w}$ for word size $w$); any cell can be read in $O(1)$ time, and arithmetic and comparisons on words cost $O(1)$. The size of an instance is its number of cells $n$; when the length in bits matters (e.g.\ for weakly NP-hard numerical problems such as Partition) we write $N$ for it. Throughout, $n$ denotes the number of cells unless stated otherwise; in case studies we keep the customary parameters (the number of variables of a formula, of cities, or of vertices), and we state the input size explicitly when it differs. Let $P$ be a computational problem, and let $I_n$ be the set of all valid input instances of size $n$.
	

\paragraph{Costs are measured for a fixed algorithm.} For an algorithm $\mathcal{A}$ that is correct on all inputs, let $T_{\mathcal{A}}(x)$ denote its running time on $x$, and write $\mathrm{Cost}_{\mathcal{A}}(P,x)=T_{\mathcal{A}}(x)$. We deliberately do \emph{not} define the cost of an instance as the infimum, over all correct algorithms, of the time spent on that instance. That per-instance infimum is degenerate: for every fixed $x$, the algorithm ``if the input equals $x$, output the stored answer; otherwise run a general solver'' is a single uniform algorithm, correct on all inputs of all sizes, that spends only $O(|x|)$ time on $x$. Hence the per-instance infimum is $O(n)$ for \emph{every} instance of \emph{every} decidable problem, whatever its general complexity, and every maximum or variation built from it is trivial. Taking the infimum separately for each $n$ does not help either, because for a fixed $n$ an algorithm may hard-wire the answers for all of the finitely many inputs of that length; and by Blum's speedup theorem \cite{blum1967machine} some problems have no best algorithm at all. The standard remedy, which we adopt, is to quantify over algorithms \emph{once}, before any instance is chosen, and to state bounds asymptotically.

\begin{definition}[Complexity under a property]
\label{def:complexityUnderK}
Let $K$ be an input property and $f:\mathbb{N}\to\mathbb{N}$. We say that $P$ is \emph{solvable in time $O(f)$ under $K$} if there exist a single algorithm $\mathcal{A}$, correct on all inputs satisfying $K$, and constants $c,n_0$ such that for all $n\ge n_0$
\[
W_{\mathcal{A}}(n,K)\;:=\;\max_{x\in I_n,\;x\models K} T_{\mathcal{A}}(x)\;\le\; c\,f(n).
\]
For $K=\epsilon$ (no knowledge) this is the classical worst-case complexity of $P$.
\end{definition}

\noindent Hard-wiring finitely many instances into $\mathcal{A}$ now changes only constants, and padding an algorithm with gratuitous steps is irrelevant: upper bounds are existence claims (``some $\mathcal{A}$ is fast''), and lower bounds have the form ``for every $\mathcal{A}$, $W_{\mathcal{A}}(n,K)$ is not $O(f(n))$''.

\begin{definition} \label{def:UnrestrictedCostVariation}
The \textit{(unrestricted) cost variation} of an algorithm $\mathcal{A}$ is the difference between its worst-case and best-case cost on instances of size $n$:
\begin{equation*}
\Delta_{\mathcal{A}}(n) = \max_{x \in I_n} T_{\mathcal{A}}(x) - \min_{x \in I_n} T_{\mathcal{A}}(x).
\end{equation*}
\end{definition}

\noindent For an algorithm that reads its whole input the minimum is at least $n$, so the substance of $\Delta_{\mathcal{A}}(n)$ lies in the maximum term, i.e.\ in $W_{\mathcal{A}}(n,\epsilon)$. When we write $\Delta_P$ or $W_P$ without naming an algorithm, we refer to a fixed algorithm stated in the text (typically the best one known).

Because $\Delta_{\mathcal{A}}$ is a difference of step counts, it depends on constant factors and machine details; for asymptotic statements we use a ratio instead, and classify algorithms and problems by their sensitivity as follows:

\begin{definition}[Not Sensitive to Uncertainty - NSU]
The \emph{sensitivity ratio} of an algorithm $\mathcal{A}$ is $\sigma_{\mathcal{A}}(n)=\max_{x\in I_n}T_{\mathcal{A}}(x)\,/\,\min_{x\in I_n}T_{\mathcal{A}}(x)$. The algorithm is \emph{Not Sensitive to Uncertainty (NSU)} if $\sigma_{\mathcal{A}}(n)=O(1)$, i.e.\ its cost is determined by the input size up to a constant factor, and \emph{exactly} NSU if $\Delta_{\mathcal{A}}(n)=0$ for every $n$. A problem $P$ is NSU if there is a function $f$ such that (i) some algorithm $\mathcal{A}$ for $P$ has $W_{\mathcal{A}}(n,\epsilon)=O(f(n))$, and (ii) every algorithm $\mathcal{B}$ that solves $P$ satisfies $\min_{x\in I_n}T_{\mathcal{B}}(x)=\Omega(f(n))$. Otherwise $P$ is \emph{sensitive to uncertainty (SU)}.
\end{definition}

In words, $P$ is NSU if \emph{every} instance of size $n$ is, for \emph{every} algorithm, as hard as the worst case up to a constant factor. If $P$ is NSU with function $f$, then every algorithm $\mathcal{A}$ with $W_{\mathcal{A}}(n,\epsilon)=O(f(n))$ is NSU, since its minimum is $\Omega(f(n))$ and its maximum is $O(f(n))$.

\begin{remark}[Why the definition quantifies over all algorithms]
\label{rem:nsu-padding}
A weaker definition, ``$P$ admits an NSU algorithm that is worst-case optimal'', would be trivial. Given any algorithm $\mathcal{A}$ whose worst-case time $f(n)$ is time-constructible, the padded algorithm that runs $\mathcal{A}$ and then idles until $f(n)$ steps have elapsed is still correct and worst-case optimal, and has $\sigma=1$; so essentially every problem would be NSU. Requiring that \emph{all} worst-case-optimal algorithms be NSU avoids padding, but becomes vacuous for problems that have no worst-case-optimal algorithm, which exist by Blum's speedup theorem \cite{blum1967machine}. Condition~(ii) avoids both problems: it is a lower bound on the best case of every algorithm, so padding cannot help, and it does not presuppose that an optimal algorithm exists.
\end{remark}

\begin{remark}[How large can the sensitivity ratio be?]
\label{rem:sigma-range}
Always $\sigma_{\mathcal{A}}(n)\ge 1$, with equality iff $\Delta_{\mathcal{A}}(n)=0$, and in general $\Delta_{\mathcal{A}}(n)=\min_{x}T_{\mathcal{A}}(x)\cdot(\sigma_{\mathcal{A}}(n)-1)$. Since every run takes at least one step, $\sigma_{\mathcal{A}}(n)\le W_{\mathcal{A}}(n,\epsilon)$; if $\mathcal{A}$ reads its whole input on every instance, $\sigma_{\mathcal{A}}(n)\le W_{\mathcal{A}}(n,\epsilon)/n$. There is no universal upper bound, because by the time hierarchy theorem $W_{\mathcal{A}}$ can be any time-constructible function. For natural algorithms, $\sigma$ ranges from constant (summation, $\sigma=1$) through polynomial (insertion sort, $\sim n$ between sorted and reversed input) to exponential: exhaustive search for Partition with a fast path for all-equal inputs has best case $\sim n$ and worst case $\sim 2^{n}n$, and trial-division factoring of $n$-bit integers has best case $O(n)$ (even $N$) and worst case $\sim 2^{n/2}$ (balanced semiprimes).
\end{remark}

To grasp the idea, let's start from a trivial task.

\begin{lemma}
The problem of summing $n$ integers (SUM) is Not Sensitive to Uncertainty (NSU).
\end{lemma}

\begin{proof}
Take $f(n)=n$. (i) The standard algorithm reads the $n$ cells and performs exactly $n-1$ additions, so in the unit-cost word RAM its running time is $c\cdot n$ for some constant $c$, independently of the values; hence its cost variation is $c\, n-c\, n=0$ for all $n$ (it is exactly NSU, with $\sigma=1$). (ii) Every correct algorithm $\mathcal{B}$ reads every cell on \emph{every} input: if $\mathcal{B}$ did not read cell $i$ on input $x$, then on the input obtained from $x$ by changing cell $i$ (which changes the sum), $\mathcal{B}$ would perform exactly the same computation and give the same, now wrong, answer. Hence $\min_{x\in I_n}T_{\mathcal{B}}(x)\ge n$. (In the bit model the cost of an addition depends on the lengths of the operands, and SUM is NSU only among inputs of equal bit length.)
\end{proof}

Conversely, many hard problems are highly sensitive to uncertainty. Their worst-case complexity is exponential (under our current understanding), while their best-case complexity (on highly structured inputs) can be linear or even constant. This variation represents a spectrum of difficulty that a single worst-case bound does not capture. The following observation makes this precise.

\begin{proposition}[Fast paths imply sensitivity]
\label{prop:fastpath}
Let $P$ be a problem such that no algorithm solves it with $W_{\mathcal{A}}(n,\epsilon)=O(n)$. Suppose there is a set $E$ of instances, containing at least one instance of every size, such that some algorithm decides in $O(n)$ time whether $x\in E$ and, if so, outputs the answer of $P$ on $x$. Then $P$ is sensitive to uncertainty.
\end{proposition}

\begin{proof}
Suppose $P$ were NSU with function $f$, and let $\mathcal{A}$ be an algorithm with $W_{\mathcal{A}}(n,\epsilon)=O(f(n))$, which exists by condition~(i). The algorithm $\mathcal{A}'$ that first tests whether $x\in E$, answers directly if so, and otherwise runs $\mathcal{A}$, is correct and spends $O(n)$ time on the instances of $E$. By condition~(ii) applied to $\mathcal{A}'$, $f(n)=O(\min_{x}T_{\mathcal{A}'}(x))=O(n)$, so $W_{\mathcal{A}}(n,\epsilon)=O(n)$, contradicting the assumption on $P$.
\end{proof}

\noindent For example, Partition (with $E$ the all-equal instances) and Integer Factorization (with $E$ the even integers) are sensitive to uncertainty unless they can be solved in linear time.

\subsection{Case Study: The Partition Problem}
\label{sec:ThePartitionProblem}

The Partition Problem is a classic $NP$-complete problem \cite{mertens2006easiest} stated as follows:
\begin{problem}[Partition Problem]
We have a finite multiset (bag) of positive integers $S = \{s_1, s_2, \dots, s_n\}$. Is there any subset $S' \subseteq S$ such that:
\begin{equation*}
\sum_{s \in S'} s = \sum_{s \in S \setminus S'} s = \frac{1}{2} \sum_{s \in S} s
\end{equation*}
\end{problem}

In the general case (where we have no extra knowledge about the input), solving the partition problem is $NP$-complete. It can be solved in $O(2^n)$ time via exhaustive search, or $O(n \cdot B)$ time via dynamic programming, where $B = \sum_{s \in S} s$. Note that the dynamic programming solution is pseudo-polynomial, meaning it becomes extremely expensive when the integers are large.
	In general, although there are several algorithms for the partition problem with time complexity less than $O(2^n)$,
	e.g. 	\cite{horowitz1974computing,schroeppel1981algorithm,howgrave2010new},
	all these are exponential,
	and for reasons of simplicity (and without loss of generality),
	hereafter  we will assume that the time complexity of the Partition problem is in $O(2^n)$.

However, if we possess extra knowledge about the input,
expressed as a {\em property} or condition $K$ that the input satisfies,
the computational cost can collapse. Below we analyze five specific cases based on different levels of prior knowledge, demonstrating how a ``KNOWING'' clause (K-clause) reduces complexity:

\begin{itemize}
\item \textit{Case 1: Two Integers ($K: |S| = 2$)}. If we know a priori that the multiset contains exactly two integers, the problem can be solved with a single comparison operator: check if $s_1 = s_2$. If true, return \texttt{True}; else return \texttt{False}. 
Therefore $\text{Cost} = O(1)$. Since $|S|=2$ fixes the size of the input, this case says nothing about growth in $n$; we include it only as the simplest example of knowledge that reduces the problem to a single comparison.
The cost to verify the K-clause itself ($|S| = 2$) requires scanning the list once to verify its size, which takes $O(n)$ time.

\item \textit{Case 2: Same Integers ($K: s_i = k$ for all $i$)}.
If we know that all elements in the multiset are identical, we only need to check if the number of elements is even: $|S| \bmod 2 = 0$. If true, we can partition them into two groups of equal size; else, an equal partition is impossible.
Therefore, $\text{Cost} = O(1)$.
Again, the verification of the K-clause (that all elements are equal to some $k$) requires $O(n)$ time to scan and compare elements.

\item \textit{Case 3: Arithmetic Sequence ($K: s_i = i$ for all $i$)}. 
If we know that the elements are the first $n$ consecutive integers $S = \{1, 2, \dots, n\}$, we can determine if a partition exists in constant time. The sum of the first $n$ integers is $H = n(n+1)/2$. A partition is possible if and only if $H$ is even, which occurs if and only if: 
$n \bmod 4 = 0 \quad \text{or} \quad (n + 1) \bmod 4 = 0$.
Thus, we can solve the problem with a simple modulo operation,
i.e. $\text{Cost} = O(1)$.
Verifying the K-clause takes $O(n)$ time: a single scan if the elements are given in order ($s_i=i$), or a scan with a bit array of size $n$ if $S$ is given as an unordered multiset.

\item \textit{Case 4: Equal Elements modulo 3 ($K: |S| \bmod 3 = 0$ and $s_i = k$)}. 
If we know that all elements are identical ($s_i = k > 0$) and their total count is a multiple of $3$, then $|S| = 3m$ for some integer $m$. The sum of elements is $3mk$. A partition into two equal halves requires each half to sum to $1.5mk$. Since $s_i = k$, this would require selecting $1.5m$ elements. If $m$ is odd, $1.5m$ is not an integer, meaning we would need to split a physical element, which is impossible. Thus, if $m$ is odd (i.e., $|S|$ is an odd multiple of 3), we can immediately return \texttt{False} in constant time,
therefore $\text{Cost} = O(1)$.
This case is subsumed by Case~2: $|S|=3m$ is odd exactly when $m$ is odd, so the answer is again determined by the parity of $|S|$. We keep it to illustrate that different formulations of knowledge can lead to the same collapse.

\item \textit{Case 5: No Knowledge ($K = \epsilon$)}.
This corresponds to the general case where we have no prior structural knowledge. The worst-case complexity remains exponential,
i.e. $\text{Cost} = O(2^n)$.
\end{itemize}

In summary, the availability of structural knowledge allows us to solve specific instances of an $NP$-complete problem in constant time $O(1)$ once the property is known to hold (the \emph{promise} setting of \S\ref{sec:MeasuringSU}), which is entirely independent of $n$. Crucially, verifying whether the input satisfies the property $K$ (the cost of the K-clause verification) 
typically 
takes $O(n)$ time.  Therefore, if the K-clause is verified, the total cost of solving the instance is $O(n) + O(1) = O(n)$, bypassing the exponential worst-case bottleneck.
For Cases 1--4 this check is a single linear scan. This is not a universal law, however: other structural predicates may be
costlier to verify. For instance, recognizing perfect graphs takes $O(|V|^9)$ time \cite{chudnovsky2005recognizing}, although on perfect graphs maximum clique and minimum coloring are solvable in polynomial time \cite{grotschel1988geometric}; and checking that a candidate set $B$ is a strong backdoor of a formula takes time exponential in $|B|$ (\S\ref{sec:MoreCollapse}). By contrast, the $\Theta(n^2)$ check of the Monge property of an $n\times n$ distance matrix (\S\ref{sec:TSPwithMonge}) is linear in the input size, which is $n^2$ cells.


\section{Measuring Sensitivity to Uncertainty}
\label{sec:MeasuringSU}

To formalize how extra knowledge reduces complexity, we restrict the instances over which costs are measured. Let $K$ denote a property satisfied by the problem instance; we write $x\models K$ when $x$ satisfies $K$, and we assume throughout that $K$ is decidable.

\begin{definition}
For an algorithm $\mathcal{A}$, the \emph{cost variation under $K$} is
\begin{equation*}
\Delta_{\mathcal{A}}(n, K) = \max_{x \in I_n, x \models K} T_{\mathcal{A}}(x) - \min_{x \in I_n, x \models K} T_{\mathcal{A}}(x),
\end{equation*}
and its maximum term is the restricted worst-case cost $W_{\mathcal{A}}(n,K)$ of Definition~\ref{def:complexityUnderK}.
\end{definition}

We denote the absence of knowledge by $\epsilon$, so that $\Delta_{\mathcal{A}}(n, \epsilon) = \Delta_{\mathcal{A}}(n)$ is the unrestricted cost variation of Section~\ref{sec:SU}. Two settings must be distinguished, because they differ exactly by the cost of checking $K$.

\begin{definition}[Promise and guarded cost]
\label{def:promiseguarded}
In the \emph{promise} setting an algorithm $\mathcal{A}$ may assume that $x\models K$, and its behavior on other inputs is unconstrained; its \emph{promise cost} $W^{\mathrm{prom}}_{\mathcal{A}}(n,K)$ is its restricted worst-case cost $W_{\mathcal{A}}(n,K)$. In the \emph{guarded} setting the algorithm must first check whether $x\models K$ and, if the check fails, reject or fall back to a general solver; the guarded algorithm $\mathcal{A}^{\mathrm{guard}}$ that runs the check and then $\mathcal{A}$ has \emph{guarded cost}
\[
W^{\mathrm{guard}}_{\mathcal{A}}(n,K)\;\le\; C_{\mathrm{verify}}(n,K)+W^{\mathrm{prom}}_{\mathcal{A}}(n,K),
\]
where $C_{\mathrm{verify}}(n,K)$ is the worst-case running time of a fixed decision procedure for $K$ on inputs of size $n$. We abbreviate $C_{\mathrm{verify}}(n,K)$ to $C_{\mathrm{verify}}(K)$ when $n$ is clear from the context, and we drop the subscript $\mathcal{A}$ when the algorithm is clear.
\end{definition}

For the Partition problem and the algorithms of Section~\ref{sec:ThePartitionProblem} we have:
\begin{itemize}
\item $W^{\mathrm{prom}}(n,|S|=2)=O(1)$, \ $W^{\mathrm{prom}}(n,s_i=k\ \forall i)=O(1)$ and $W^{\mathrm{prom}}(n,s_i=i\ \forall i)=O(1)$;
\item the corresponding guarded costs are $O(n)$, dominated by the check;
\item for $K=\epsilon$ the best known worst-case bounds are exponential ($O^{*}(2^{n/2})$, by the algorithms of Horowitz--Sahni and Schroeppel--Shamir \cite{horowitz1974computing,schroeppel1981algorithm}), and the cost variation of these algorithms is correspondingly exponential.
\end{itemize}
Statements of the form ``$\Delta=0$'' depend on the algorithm: the parity test under $s_i=k$ has zero variation, whereas an algorithm with early exits does not.

The maximum term is the \emph{worst-case complexity of $P$ restricted to $\{x\models K\}$}:
\begin{equation}
\label{eq:max-def}
W_{\mathcal{A}}(n,K)\;=\;\max_{x\in I_n,\;x\models K} T_{\mathcal{A}}(x),
\end{equation}
and for $K=\epsilon$ it is the classical worst-case complexity, so the framework subsumes the standard analysis as the special case $K=\epsilon$. What the framework adds is the explicit dependence on $K$: as structural knowledge increases, the restricted worst case can drop from exponential (Partition) to polynomial or linear time. We call a problem \emph{knowledge-vulnerable} with respect to a family of properties if some property of the family that covers a substantial part of the instance space (Section~\ref{sec:Parametric}) makes it easy, and \emph{knowledge-resistant} if no such property does. Section~\ref{sec:ReallyHard} makes this precise and explains why Integer Factorization is the canonical, conditional, example of knowledge resistance, even though many specific properties do make factoring easy. This dependence on $K$ is the axis that $K$-free worst-case statements do not express.

\subsection{Characterizing Input Uncertainty: The Metric $\Psi(n, K)$}
\label{sec:ToPsi}

We introduce a metric $\Psi(n, K)$ to quantify the amount of unknown information remaining in the input once property $K$ is known. Let the input be the array $A[1..n]$ of cells of \S\ref{sec:SU}. For $x\models K$, a set $Q\subseteq\{1,\dots,n\}$ \emph{determines} $x$ under $K$ if every $y\in I_n$ with $y\models K$ and $y_Q=x_Q$ equals $x$; that is, knowing $K$ and the values $A[i]_{i\in Q}$ uniquely determines $A[1..n]$. We define $\Psi(n,K)$ as the number of cells that must be read, in the worst case over instances, in order to eliminate all uncertainty and reconstruct the input:
\begin{equation*}
\Psi(n, K) \;=\; \max_{x\in I_n,\;x\models K}\;\min \bigl\{\, |Q| \;:\; Q \subseteq \{1,\dots,n\} \text{ determines } x \text{ under } K \,\bigr\}.
\end{equation*}
A determining set is \emph{canonical} if a single set $Q$ of size $\Psi(n,K)$ determines every $x\models K$ of size $n$.

Here the cells range over the cell alphabet $\Sigma$ of the computational model (\S\ref{sec:SU}); for Partition a cell holds one integer. Note that $\Psi$ concerns the reconstruction of the \emph{input}; the number of cells needed to \emph{answer} the question can be much smaller (the task-relative variant $\Psi_{\mathrm{task}}$ of Remark~\ref{rem:recon-vs-task}).

Let us evaluate $\Psi(n, K)$ for the Partition cases analyzed previously:
\begin{itemize}
\item $\Psi(n, |S| = 2) = 2$. We only need to read 2 cells to know the full input.
\item $\Psi(n, \text{``allEqual''}) = 1$. We only need to read 1 cell to know the value $k$; all other cells are known to be equal to $k$.
\item $\Psi(n, \text{``}s_i = i\text{''}) = 0$. We need to read 0 cells because the input is completely determined by the arithmetic sequence property.
\item $\Psi(n, \epsilon) = n$. In the absence of knowledge, we must read all $n$ cells to understand the input.
\end{itemize}

\noindent
Note that the definition of $\Psi(n,K)$ is {\bf existential}: it asserts the
\emph{existence} of a determining set $Q$ of minimum cardinality, but it does
not guarantee that an algorithm can \emph{locate} such a set efficiently. In
general, discovering which cells form a determining set is itself a
computational problem with a non-trivial cost, which we formalize as the
\emph{find cost}.

\begin{definition}[Find cost]
\label{def:findcost}
For a property $K$, the \emph{find cost} $C_{\mathrm{find}}(n,K)$ is the worst-case running time, over instances $x\models K$ of size $n$, of a fixed procedure $A$ that outputs a determining set $Q\subseteq\{1,\dots,n\}$ of $x$ with $|Q|\le\Psi(n,K)$:
\[
C_{\mathrm{find}}(n,K)\;=\;\max_{x\in I_n,\;x\models K}\;T_{A}(x),
\]
where $T_A(x)$ is the running time of $A$ on instance $x$ until it returns a
determining set. As with costs in general (\S\ref{sec:SU}), we do not minimize over procedures separately for each $n$, which would be degenerate; statements about $C_{\mathrm{find}}$ are asymptotic statements about a fixed procedure, and a ``lower bound'' on $C_{\mathrm{find}}$ means that no procedure achieves a smaller asymptotic bound. When $Q$ can be determined immediately from the property
itself, e.g.\ for $|S|=2$ any two cells suffice, for ``all equal'' any
single cell suffices, for $s_i=i$ no cell is needed, we say that $Q$ is
\emph{visible} and $C_{\mathrm{find}}(n,K)=O(1)$.
\end{definition}

\begin{example}[A non-visible determining set]
\label{par:tree}
To see that $C_{\mathrm{find}}$ is not always $O(1)$, consider the property $K_{\mathrm{tree}}$: the input is the $n\times n$ matrix $D$ of pairwise distances between the $n$ vertices of an (unknown) tree $T$ with positive edge weights. Let $m=n^{2}$ denote the number of matrix cells, i.e.\ the input size. The $n-1$ cells $D[u,v]$ such that $\{u,v\}$ is an edge of $T$ form a determining set: their positions give the topology of $T$, their values give the edge weights, and every other entry is the sum of the edge weights along a path. Hence $\Psi(m,K_{\mathrm{tree}})\le n-1$. However, \emph{which} cells these are is not visible from $K$: finding them is the tree-reconstruction problem (a cell $D[u,v]$ corresponds to an edge iff there is no $w\notin\{u,v\}$ with $D[u,w]+D[w,v]=D[u,v]$), which takes polynomial time, between $O(n^{2})$ and $O(n^{3})$ depending on the method \cite{Buneman1971Recovery}.
\end{example}

\subsubsection{Variants of $\Psi(n,K)$}

Let's now define  $\Psi_{\mathrm{acc}}$, the {\bf constructive} counterpart of the {\em existential} $\Psi(n,K)$.

\begin{definition}[Accessible uncertainty]
	\label{def:psiac}
	$\Psi_{\mathrm{acc}}(n,K)$ is the least $s$ such that some fixed polynomial-time procedure outputs, for every $x\models K$ of size $n$, a determining set of $x$ of size at most $s$.
\end{definition}

Clearly $\Psi_{\mathrm{acc}}(n,K)\ge \Psi(n,K)$, with equality whenever determining sets of minimum size can be computed in polynomial time, e.g.\ when $Q$ is visible.

\paragraph{The prefix variant $\Psi_{\mathrm{pref}}$.}
\label{par:pref}
Now suppose we require the determining set to be a \emph{prefix} of the tape,
i.e.\ $Q=\{1,\dots,m\}$ for some $m$. Then the uncertainty-vanishing metric
becomes:
\[
\Psi_{\mathrm{pref}}(n,K)\;=\;\min\bigl\{\,m\;:\;
\text{knowing } K \text{ and } A[1..m] \text{ uniquely determine } A[1..n]\,\bigr\},
\]
where the condition must hold for every $x\models K$ of size $n$.
Obviously it holds $\Psi(n,K)\le \Psi_{\mathrm{pref}}(n,K)\le n$,
and reaching the lower bound exactly characterizes when the minimal residual
information happens to sit at the beginning of the tape.
This variant is \emph{syntactic} and \emph{operational}. The gain is that
$Q$ is known a priori, hence $C_{\mathrm{find}}(n,K)=O(1)$ always
and  the
Extraction stage needs only a sequential tape (no RAM-TM).


\subsection{Connection to Information Theory}

The metric $\Psi(n, K)$ is directly related to conditional Kolmogorov complexity
\footnote{The conditional Kolmogorov complexity $\kappa (x \mid K)$ of a string $x$ given an auxiliary description (or structural property) $K$ is defined as the length of the shortest binary program $p$ that, when executed on a reference universal Turing machine $U$ with $K$ provided as an auxiliary input, outputs $x$. Formally, $\kappa_U(x \mid K) = \min \{ |p| : U(p, K) = x \}$. In our framework, this measures the absolute minimum computational information required to completely specify the input instance $x$ once the structural constraints guaranteed by $K$ are already known and utilized.}
 \cite{li2008introduction,l_and_v_2019} and Shannon entropy \cite{Shannon1948AM}. The information that remains unknown about an input $x$ of size $n$ once $K$ is known can be measured by the conditional Kolmogorov complexity $\kappa(x\mid K,n)$ (defined like $\kappa(x\mid K)$, with $n$ also given to the program), and $\Psi$ bounds it:
 

\begin{lemma}[Compression Bound]
	\label{lem:compression}
	    Let $x$ be an input string of size $n$ satisfying a property $K$, let the cell alphabet $\Sigma$ be finite, and assume
	the active set $Q$ is \emph{canonical}, i.e.\ uniquely determined by $K$ and
	$n$ (no per-instance choice of which cells to read). Then the remaining
	uncertainty of $x$ given $K$ is bounded by its uncertainty-vanishing metric,
	i.e.,
\[
	\kappa(x \mid K, n) \;\le\; \Psi(n,K)\cdot \log|\Sigma| + O(1),
\]
	where $\Sigma$ is the cell alphabet.
\end{lemma}

\begin{proof}
	By definition of $\Psi(n,K)$, there exists a set of positions
	$Q \subseteq \{1,\dots,n\}$ with $|Q| = \Psi(n,K)$ such that the property $K$,
	together with the cell values $A[i]_{i\in Q}$, \emph{uniquely determines} the
	entire input $A[1..n]$. Because $K$ is supplied as the auxiliary input, it fixes
	both the structural constraint and the reconstruction rule that completes the
	remaining $n - \Psi(n,K)$ cells.
	
	Consider the program $p$ that, given $K$ and $n$ as auxiliary input, performs the
	following three steps:
	(1) reads the values $A[i]_{i\in Q}$ on the active cells,
	(2) applies the fixed reconstruction procedure determined by $K$ to fill in
		the remaining cells, and 
	(3) outputs the resulting string $x$.
	Each active cell holds a symbol of the alphabet $\Sigma$, hence it can be encoded
	in $\log|\Sigma|$ bits. The program $p$ therefore needs only to carry the active
	data, of length $\Psi(n,K)\cdot\log|\Sigma|$ bits, together with a constant-size
	skeleton (the reconstruction step and the addressing of $Q$ are already
	determined by $K$). Consequently,
$
	\kappa_U(x \mid K) \;\le\; |p| \;=\; \Psi(n,K)\cdot\log|\Sigma| + c
	\;=\; \Psi(n,K)\cdot\log|\Sigma| + O(1),
$
	and since $\kappa(x\mid K)$ is the \emph{minimum} length over all such programs, the
	claimed bound follows.
\end{proof}

The same argument gives a counting (Hartley-entropy) version that needs no reference to Kolmogorov complexity: if the determining set is canonical, then
\begin{equation}
\label{eq:counting}
\bigl|\{x\in I_n : x\models K\}\bigr| \;\le\; |\Sigma|^{\Psi(n,K)},
\end{equation}
because $K$ and the $\Psi(n,K)$ values on $Q$ determine $x$; without canonicity, each $x$ is determined by a pair $(Q,x_Q)$ with $|Q|\le\Psi(n,K)$, so the bound becomes $\sum_{j\le\Psi}\binom{n}{j}|\Sigma|^{j}\le(\Psi+1)\,(n|\Sigma|)^{\Psi}$, with $\Psi=\Psi(n,K)$. In this light, a structural property $K$ acts as a lossless compression of the input space. Bound \eqref{eq:counting} also limits how hard a problem restricted to a property with small $\Psi$ can be (Proposition~\ref{prop:smallpsi}).

\begin{remark}
	The $O(1)$ overhead in Lemma~\ref{lem:compression} relies on the active set $Q$
	being canonical, i.e.\ fixed by $K$ and $n$. In full generality, if $Q$ may
	depend on the instance, specifying it costs an additional
	$\log\binom{n}{\Psi(n,K)}$ bits
	(because the decoder needs to be told which subset was used,
	and there are 
	$\binom{n}{\Psi(n,K)}$  candidates), yielding
	$\kappa(x\mid K,n) \le \Psi(n,K)\log|\Sigma| + \log\binom{n}{\Psi(n,K)} + O(\log n)$, where the $O(\log n)$ term also covers encoding $|Q|$.
Note that this quantifier issue is the information-theoretic counterpart of the find
cost of Section~\ref{sec:ToPsi}: the additional
$\log\binom{n}{\Psi(n,K)}$ bits are what it costs to \emph{describe} which
subset $Q$ was used, whereas $C_{\mathrm{find}}(n,K)$ is what it costs to
\emph{compute} it. 
\end{remark}

\section{\texorpdfstring{The Uncertainty Reduction Theorem}{The Uncertainty Reduction Theorem}}
\label{sec:UncertaintyReductionTheorem}

We now relate the guarded cost under $K$ to the residual uncertainty $\Psi(n,K)$. Once $K$ and the values on a determining set $Q$ are known, the input is fixed, so the remaining computation is a problem on $(n,Q,x_Q)$. This \emph{reduced problem} must be defined explicitly, because in general it is \emph{not} an instance of $P$.

\begin{definition}[Reduced problem]
\label{def:reduced}
Fix a property $K$. For $x\models K$ of size $n$ and a determining set $Q$ of $x$, let $x_Q\in\Sigma^{|Q|}$ be the projection of $x$ onto the cells of $Q$, and let $R_{n,K}(Q,x_Q)=x$ be the reconstruction map, which is well defined by the definition of a determining set. The \emph{reduced problem} $P_K$ takes as input a triple $(n,Q,y)$ for which $R_{n,K}(Q,y)$ is defined, and asks for the answer of $P$ on $R_{n,K}(Q,y)$.
\end{definition}

For example, under $K_{\mathrm{allEqual}}$ the reduced instance of Partition is $(n,\{1\},k)$, and its answer is ``$n$ is even''. This is \emph{not} the answer of Partition on the one-element multiset $\{k\}$ (which is always ``no''): the answer depends on $n$, which is why $n$ is part of the reduced input.

\begin{theorem}[Uncertainty Reduction]
\label{thm:uncertainty}
Let $P$ be a computational problem and $K$ a decidable property. Suppose that
(i)~$K$ can be checked in time $C_{\mathrm{verify}}(n,K)$,
(ii)~for every $x\models K$ of size $n$, a determining set $Q$ of $x$ with $|Q|\le\Psi(n,K)$ can be computed in time $C_{\mathrm{find}}(n,K)$, and
(iii)~an algorithm $\mathcal{B}$ solves the reduced problem $P_K$ in time $T_{\mathcal{B}}(n,m)$ on inputs with $|Q|=m$, where $T_{\mathcal{B}}$ is non-decreasing in $m$.
Let $\mathcal{A}_K$ be the guarded algorithm that checks $K$ (rejecting, or falling back to a general solver, if the check fails), computes $Q$, reads $x_Q$, and runs $\mathcal{B}$ on $(n,Q,x_Q)$. Then $\mathcal{A}_K$ is correct, and for every $x\in I_n$ with $x\models K$,
\begin{equation}
\label{eq:urd-constructive}
T_{\mathcal{A}_K}(x)\;\le\;
C_{\mathrm{verify}}(n,K)+C_{\mathrm{find}}(n,K)+O\bigl(\Psi(n,K)\bigr)+T_{\mathcal{B}}\bigl(n,\Psi(n,K)\bigr).
\end{equation}
Consequently, the guarded cost $W^{\mathrm{guard}}(n,K)\le W_{\mathcal{A}_K}(n,K)$ and the variation $\Delta_{\mathcal{A}_K}(n,K)$ are bounded by the right-hand side of \eqref{eq:urd-constructive}; in the promise setting the term $C_{\mathrm{verify}}(n,K)$ disappears.
\end{theorem}

\begin{proof}
Correctness: for $x\models K$ the check succeeds and, since $Q$ is a determining set, $R_{n,K}(Q,x_Q)=x$; hence the answer of $\mathcal{B}$ on $(n,Q,x_Q)$ is the answer of $P$ on $x$. Inputs that fail the check are handled by the fallback. Running time: the four stages (verification, locating, reading and execution) take at most $C_{\mathrm{verify}}(n,K)$, $C_{\mathrm{find}}(n,K)$, $O(\Psi(n,K))$ (random-access reads) and $T_{\mathcal{B}}(n,\Psi(n,K))$ steps, respectively. Taking the maximum over $x\models K$ bounds $W_{\mathcal{A}_K}(n,K)$, and $\Delta_{\mathcal{A}_K}(n,K)\le W_{\mathcal{A}_K}(n,K)$ because running times are non-negative.
\end{proof}

The theorem is elementary; its value lies in the bookkeeping it enforces. Each of the four terms is a genuine cost that can dominate: the check (\S\ref{sec:ThePartitionProblem}), the location of the residual information (Example~\ref{par:tree}), the reading, and the reduced problem. In particular, the last term depends on $n$ as well as on $\Psi$, and small $\Psi$ does not by itself make it small (Proposition~\ref{prop:smallpsi}).

\begin{corollary}[Projection regime]
\label{cor:projection}
Suppose, in addition, that there is a map $\rho$, computable in time $t_\rho(n)$, that turns every reduced instance $(n,Q,y)$ into an instance of $P$ itself, of size at most $s(\Psi(n,K))$ and with the same answer, and let $\mathcal{S}$ be any algorithm for $P$. Then
$
T_{\mathcal{B}}(n,\Psi(n,K))\;\le\; t_\rho(n)+W_{\mathcal{S}}\bigl(s(\Psi(n,K)),\epsilon\bigr).
$
\end{corollary}

\noindent When $\Psi$ is bounded by a parameter $k$ of the instance, a map $\rho$ as in Corollary~\ref{cor:projection} plays the role of a \emph{kernelization} with kernel size $s(k)$ \cite{fomin2019kernelization} for the problem restricted to $K$; it is a kernel in the usual sense only if it is defined on all instances and also bounds the parameter of its output. For example, under $K_{\mathrm{allEqual}}$ for Partition, $\rho$ outputs the instance $(k,k)$ if $n$ is even and $(k)$ otherwise, of constant size. A related reduction, measured in variables rather than cells, applies to the property $K_f$ of \S\ref{sec:ETH}: simplifying a 3-CNF formula by the fixed part of the assignment yields an equivalent formula over the at most $f(n)$ free variables.

\ \\
\noindent
{\bf The variant $\Psi_{\mathrm{task}}$.}

	\begin{remark}[Reconstruction vs.\ task resolution]
		\label{rem:recon-vs-task}
		The metric $\Psi(n,K)$ measures the minimum number of cells needed to
		\emph{reconstruct the full input} given $K$. The Uncertainty Reduction
		Theorem uses this as an \emph{upper bound} on the information the solver
		requires: one can always reconstruct the input from the
		$\Psi(n,K)$ determining cells and then solve $P$ on the reconstruction.
		However, the solver for $P$ need not distinguish between inputs that
		yield the same answer. For example, under $K_{\mathrm{allEqual}}$,
		$\Psi(n,K)=1$ (one cell suffices to determine the common value), yet
		the decision version of Partition depends only on the parity of $n$, so the
		common value need not be read at all.
		Many distinct inputs (e.g.\ $(1,1,\ldots,1)$ and $(8,8,\ldots,8)$)
		produce the same answer, and the algorithm need not distinguish them.
		We can therefore define the {\em task-relative variant}, the (adaptive) cell-probe complexity of $P$ under the promise $K$:
		\[
		\Psi_{\mathrm{task}}(n,K,P) \;=\; \min_{\mathcal{D}}\;\max_{x\in I_n,\,x\models K}\;\#\{\text{cells read by } \mathcal{D} \text{ on } x\},
		\]
		where $\mathcal{D}$ ranges over adaptive probing strategies (decision trees) that determine the answer of $P$ on every $x\models K$.
		If the determining set is canonical, then $\Psi_{\mathrm{task}}(n,K,P)\le\Psi(n,K)$, since reading the fixed set $Q$ determines the input and hence the answer. Without canonicity this can fail, because an adaptive strategy must first \emph{locate} a determining set: under $K$ = ``exactly one cell is nonzero'', $\Psi(n,K)=1$ (the nonzero cell determines the input), yet finding that cell, and hence answering a question such as ``what is its value?'', requires up to $n-1$ probes. This gap is the cell-probe counterpart of the find cost $C_{\mathrm{find}}$. Moreover, $\Psi_{\mathrm{task}}(n,K_{\mathrm{allEqual}},\text{Partition})=0$.
		Note that   
		the bound 
		of theorem  \ref{thm:uncertainty}
		remains valid, but it should be read as an upper bound on the
		information the solver consumes, not as a lower bound on the work
		required to solve $P$. A lower bound is given by $\Psi_{\mathrm{task}}$ itself: every algorithm $\mathcal{A}$ that is correct under the promise $K$ reads at least $\Psi_{\mathrm{task}}(n,K,P)$ cells on some instance, so $W^{\mathrm{prom}}_{\mathcal{A}}(n,K)\ge\Psi_{\mathrm{task}}(n,K,P)$.
	\end{remark}

Let us now provide some refinements about the kinds of uncertainty.

\begin{definition}[Kinds of Uncertainty]
	\label{def:activedim}
	We distinguish three notions of ``remaining unknowns.''
	\begin{enumerate}
		\item The \emph{reconstruction metric} $\Psi(n,K)$ counts the minimum
		number of cells whose values, together with $K$, determine the entire
		input (\S\ref{sec:ToPsi}).
		\item The \emph{task-relative metric} $\Psi_{\mathrm{task}}(n,K,P)$, which is at most $\Psi(n,K)$ for canonical determining sets
		(Remark~\ref{rem:recon-vs-task}) counts the minimum number of cells that must be read to
		\emph{decide} $P$.
		\item The \emph{active dimension} $\alpha(P,K)$ is the number of degrees of
		freedom the \emph{core solver} branches over once $K$ is exploited, i.e.\
		the parameter in an FPT-style bound $W(n,K)\le f(\alpha)\cdot n^{O(1)}$.
	\end{enumerate}
	Note that 
	$\alpha(P,K)$ is a property of the \emph{solver}, not of the input, and is
	distinct from the \emph{active cells} (the $\Psi(n,K)$ determining cells read
	during extraction). Two regimes occur. In the \emph{projection regime}
	(Partition under $K_{\mathrm{allEqual}}$) the solver
	operates on the $\Psi(n,K)$-cell projection, so $\alpha(P,K)\le\Psi(n,K)$. In
	the \emph{search-interface regime} (bounded-treewidth Vertex Cover, strong
	backdoors, the free variables of $K_f$ in \S\ref{sec:ETH}, see \S  \ref{sec:MoreCollapse}) the solver reads all $n$ cells 
	($\Psi=\Theta(n)$) yet
	branches only over a structural parameter, such as the treewidth $\tau$ or the
	backdoor size $k$, so $\alpha(P,K)=\tau$ (resp.\ $k$), independent of
	$\Psi(n,K)$.
\end{definition}

\noindent
\textbf{Meaning.} Theorem~\ref{thm:uncertainty} states that, under partial knowledge $K$, the cost of a guarded solver is the cost of checking $K$, plus the cost of locating and reading the $\Psi(n,K)$ cells that remain unknown, plus the cost of a reduced problem on those cells. If $\Psi(n,K)$ is constant \emph{and} the reduced problem is easy (as in the Partition cases), the guarded cost collapses to $O\bigl(C_{\mathrm{verify}}(n,K)+C_{\mathrm{find}}(n,K)\bigr)$, which is polynomial whenever both overheads are polynomial.

\ \\
\noindent
{\bf Verification cost.}
$C_{\mathrm{verify}}(n,K)$ is a first-class cost of the framework and is not bounded by $O(n)$ in general: it depends on the predicate $K$. For the properties used in the case studies of this paper (equality, arithmetic sequence, 2-CNF, Horn, tree, bipartiteness, and the Monge property of a distance matrix) verification takes time linear in the input size (for an $n\times n$ matrix, $O(n^2)$ for $n^2$ cells). Other natural properties are costlier: recognizing perfect graphs takes $O(|V|^9)$ time \cite{chudnovsky2005recognizing}, and checking for a strong backdoor of size $k$ takes time exponential in $k$ (\S\ref{sec:MoreCollapse}). For the small-$\Psi$ properties of this paper, the verification scan also locates $Q$, so $C_{\mathrm{find}}(n,K)$ is absorbed into $C_{\mathrm{verify}}(n,K)$; the two are distinct when locating the determining set requires work beyond checking $K$, as in Example~\ref{par:tree}.

\ \\
\noindent
{\bf Finding cost.}
A second first-class cost is the find cost $C_{\mathrm{find}}(n,K)$ of Definition~\ref{def:findcost}. For the properties with small $\Psi$ used in this paper ($|S|=2$, equality, arithmetic sequence) the determining set is \emph{visible}, hence $C_{\mathrm{find}}(n,K)=O(1)$. For structural properties such as 2-CNF, Horn, tree, bipartiteness, bounded degree or planarity, $\Psi(n,K)=\Theta(n)$: the collapse comes from the active dimension rather than from reading fewer cells, and the find cost plays no role. For the tree-metric property of Example~\ref{par:tree}, $C_{\mathrm{find}}$ is polynomial ($O(n^{2})$--$O(n^{3})$), and for arbitrary properties computing a minimum determining set can be computationally hard. The constructive form of the theorem separates $C_{\mathrm{find}}$ explicitly, so that no hidden cost is swept under the existential definition of $\Psi$.

\ \\
\noindent
{\bf Relation to the classical cost.}
One might object that the framework merely re-expresses the classical cost. The new ingredient is the explicit input scale: without knowledge the information that must be read is the whole input, $\Psi(n,\epsilon)=n$; under $K$ it shrinks to $\Psi(n,K)$ (or to $\Psi_{\mathrm{task}}$), and the remaining cost is carried by the reduced problem. Parameterized complexity already lets complexity depend on a parameter besides $n$, and kernelization already reduces instances to a size that depends only on the parameter. What the present accounting adds is that the parameter is derived from a predicate $K$ on the instance, and that the costs of checking $K$ and of locating the residual information are charged explicitly rather than assumed away.

\section{Parametric Complexity Classes: $\UP[K]$ and $\UNP[K]$}
\label{sec:Parametric}

To formalize these concepts in structural complexity terms, we consider classes that depend on an input property. Throughout this section $K$ is a fixed property that is decidable in polynomial time, and we identify $K$ with the language $\{x : x\models K\}$.

\begin{definition} {\em Uncertainty-Parametric P - $\UP[K]$.}
A language $L \subseteq \Sigma^*$ belongs to the class $\UP[K]$ if there exist a deterministic Turing machine $M$ and a polynomial $p$ such that, for every string $x$ satisfying $K$, $M$ decides whether $x \in L$ within $p(|x|)$ steps.
\end{definition}

\begin{definition} {\em Uncertainty-Parametric NP - $\UNP[K]$.}
A language $L \subseteq \Sigma^*$ belongs to the class $\UNP[K]$ if there exist a non-deterministic Turing machine $M$ and a polynomial $p$ such that, for every string $x$ satisfying $K$, $M$ decides whether $x \in L$ within $p(|x|)$ steps (in the usual non-deterministic sense).
\end{definition}

In both definitions $K$ is fixed in advance and checkable in polynomial time, and $M$ is unconstrained on inputs that violate $K$. Because $K\in P$, the following equivalences are immediate:
\begin{equation}
\label{eq:promise}
L\in\UP[K] \iff L\cap K\in P \iff \bigl(L\cap K,\;K\setminus L\bigr)\ \text{is in promise-}P,
\end{equation}
and similarly $L\in\UNP[K]\iff L\cap K\in NP$. Thus $\UP[K]$ and $\UNP[K]$ are restrictions of languages to polynomial-time promises in the sense of Even, Selman and Yacobi \cite{evenselman1984promise,goldreich2006promise}, and $\UP[\epsilon]=P$, $\UNP[\epsilon]=NP$.

By defining these classes, we can say that an $NP$-complete problem like Partition belongs to $\UP[|S|=2]$ and to $\UP[\text{``allEqual''}]$. A problem's hardness thus becomes a function of the available structural knowledge $K$ rather than a single label.

\paragraph{Coverage.} A statement of the form ``$L\in\UP[K]$ for some $K$'' is by itself uninformative: every decidable $L$ belongs to $\UP[K]$ for, e.g., every finite $K$. What matters is how much of the instance space $K$ covers. When $I_n$ is finite, we define the \emph{coverage} of $K$ at size $n$ as
\[
\mu_n(K)=\frac{|\{x\in I_n: x\models K\}|}{|I_n|},
\qquad\text{or}\qquad
\mu_n^{\mathcal{D}}(K)=\Pr_{x\sim\mathcal{D}_n}[x\models K]
\]
for a distribution $\mathcal{D}=\{\mathcal{D}_n\}$ of interest. A statement $L\in\UP[K]$ is informative to the extent that $K$ is natural for the application and has non-negligible coverage; properties whose coverage tends to $1$ make $L$ \emph{generically} polynomial in the sense of generic-case complexity \cite{kapovich2003generic}.

We can now define the \textit{Uncertainty Sensitivity Index (USI)}.

\begin{definition}
The \textit{Uncertainty Sensitivity Index (USI)} of a problem $P$ with respect to a family of properties $\mathcal{K}$, fixed a priori, is the function
\begin{equation*}
\text{USI}(P, \mathcal{K})(n) = \min \{ \Psi(n, K) : K \in \mathcal{K} \text{ such that } P \in \UP[K] \}.
\end{equation*}
\end{definition}

This index measures the smallest residual uncertainty that a property of the family can leave while making the problem polynomial. The family must be fixed in advance and should exclude degenerate properties: if $\mathcal{K}$ contains singleton properties $\{x_0\}$, then $\text{USI}=0$ for every problem. A natural choice is to restrict $\mathcal{K}$ to properties whose coverage is at least some threshold $\delta(n)$.

\paragraph{What small residual uncertainty implies, and what it does not.} The counting bound \eqref{eq:counting} limits how hard a problem restricted to a property with small $\Psi$ can be, but it does not make the problem easy.

\begin{proposition}
\label{prop:smallpsi}
Let the cell alphabet $\Sigma$ have constant size, let $L\in NP$, and let $K$ be a property decidable in polynomial time.
\begin{enumerate}
\item If $\Psi(n,K)=O(1)$, or if $\Psi(n,K)=O(\log n)$ with canonical determining sets, then $L\cap K$ is sparse. Hence, by Mahaney's theorem \cite{mahaney1982sparse}, $L\cap K$ is not NP-complete (equivalently, since it is in NP, not NP-hard under polynomial-time many-one reductions) unless $P=NP$.
\item Small $\Psi$ does not imply tractability: if $E\neq NE$, there exist $L\in NP$ and $K$ with $\Psi(n,K)=0$ such that $L\notin\UP[K]$.
\end{enumerate}
\end{proposition}

\begin{proof}
(1) By \eqref{eq:counting} and the remark following it, the number of instances of size $n$ satisfying $K$ is at most $(\Psi+1)(n|\Sigma|)^{\Psi}$, which is polynomial in $n$ when $\Psi=O(1)$; with canonical determining sets it is at most $|\Sigma|^{\Psi}$, which is polynomial when $\Psi=O(\log n)$. Hence $L\cap K$, which is in NP because $K\in P$, has polynomially many strings of each length, i.e.\ it is sparse. Mahaney's theorem states that if a sparse set is NP-complete under polynomial-time many-one reductions, then $P=NP$.
(2) Book \cite{book1974tally} proved that $E=NE$ iff every tally set in NP is in P. If $E\neq NE$, let $T\subseteq\{0\}^*$ be a tally set in $NP\setminus P$, and let $L=T$ and $K=\{0\}^*$. Every $x\models K$ of length $n$ equals $0^n$, so $\Psi(n,K)=0$; and $L\in\UP[K]$ would imply $T\in P$.
\end{proof}

Part~(1) applies to properties whose instances are essentially determined by their size, such as the arithmetic sequence $s_i=i$, whose restriction of Partition has one instance per size. It does not apply to Partition under all-equal or $|S|=2$ when cells hold arbitrary integers: with $w$-bit words there are about $2^{w}$ choices for the common value, so these restricted problems are not sparse when measured in bits, even though $\Psi$ is constant. Part~(2) shows that, in Theorem~\ref{thm:uncertainty}, the reduced problem can carry all of the difficulty even when nothing remains to be read. The Partition cases are easy because their reduced problems (the parity of $n$, or $n \bmod 4$) happen to be easy, not merely because $\Psi$ is small.

\section{NP Problems Sensitive and Insensitive to Uncertainty}
\label{sec:NPProblemsSensitiveInsensitive}

\subsection{Problems Sensitive to Uncertainty}

Several fundamental computational problems in different domains exhibit high Sensitivity to Uncertainty (SU). 
Within our formalized framework, these problems are characterized by a large gap between their unrestricted worst-case cost $W(n,\epsilon)$ and their restricted worst-case cost $W(n,K)$ under specific structural properties $K$.

Below is one highly illustrative computational problem analyzed through the lens of our $(W(n, K), \Psi(n, K))$ framework.

\subsubsection{Graph Isomorphism (Planarity and Bounded Degree)}
The general Graph Isomorphism (GI) problem asks if two graphs of size $n$ are structurally identical. 
Its general worst-case complexity is quasi-polynomial, $2^{O(\log^c n)}$ (Babai’s algorithm \cite{Babai2015GraphIsomorphism,Babai2016GraphII}).
However, the complexity of GI is highly sensitive to the topological properties of the input graphs. 
Three classical structural properties illustrate this:
\begin{itemize}
	\item Property $K_{\text{planar}}$: The input graphs are planar.
			For $K_{\text{planar}}$, Hopcroft and Wong showed \cite{Hopcroft1974LinearTA} that planar graph isomorphism is solvable in {\em linear time $O(n)$}.
			Therefore 
			$
			\text{GI} \in \UP[K_{\text{planar}}]
			$.
			
	\item Property $K_{\text{tree}}$: The input graphs are acyclic and connected (trees).
	For $K_{\text{tree}}$, the tree isomorphism problem can be solved via the AHU (Aho, Hopcroft, Ullman) algorithm \cite{Aho1974TheDA} in {\em linear time $O(n)$}.
	Therefore 
	$
\text{GI} \in \UP[K_{\text{tree}}]
	$.
	
	\item Property $K_{\text{deg}}$: The maximum vertex degree is bounded by a constant $d$.
	For $K_{\text{deg}}$, Luks' group-theoretic algorithm \cite{Luks1982IsomorphismOG} solves the problem in polynomial time {\em $O(n^{O(d)})$},
	and if $d$ is constant, GI collapses to $P$. 
\end{itemize}





Table~\ref{tab:complexity_metrics} summarizes these properties together with the Partition cases, and it shows two different mechanisms. For Partition under $|S|=2$ or all-equal, the residual uncertainty $\Psi(n,K)$ is constant and the promise cost is $O(1)$. For Graph Isomorphism, in contrast, planarity, acyclicity or bounded degree leave the input essentially undetermined, $\Psi(n,K)=\Theta(n)$: the collapse is obtained not by reading less of the input but by exploiting its structure, i.e.\ through the active dimension of the search (Definition~\ref{def:activedim}).

	\begin{table}[htbp]
		\centering
		\caption{Complexity analysis and uncertainty metrics under property $K$ ($n$: number of cells; costs in the promise setting, the guarded cost adds $C_{\mathrm{verify}}$, here $O(n)$).}
		\label{tab:complexity_metrics}
		\footnotesize 
		\begin{tabular}{|p{18mm}|p{16mm}|p{22mm}|p{18mm}|p{20mm}|p{34mm}|}\hline
			\textbf{Problem $P$} & \textbf{General Cost ($K = \epsilon$)} & \textbf{Property $K$} & \textbf{Residual uncertainty $\Psi(n, K)$} & \textbf{Conditional cost} & \textbf{Mechanism / class} \\\hline
			Partition & $2^{O(n)}$ & $|S| = 2$ or all-equal & $\le 2$ & $O(1)$ & residual uncertainty; $\UP[K]$  \\\hline
			Graph Isomorphism & $2^{O(\log^c n)}$ & planar graphs & $\Theta(n)$ & $O(n)$ & structure; $\UP[K]$  \\\hline
			Graph Isomorphism & $2^{O(\log^c n)}$ & trees & $\Theta(n)$ & $O(n)$ & structure; $\UP[K]$  \\\hline
			Graph Isomorphism & $2^{O(\log^c n)}$ & maximum degree $\le d$ & $\Theta(n)$ & $n^{O(d)}$ & structure; $\UP[K]$ for constant $d$  \\\hline
		\end{tabular}
	\end{table}

\subsection{\texorpdfstring{Knowledge-Resistant Problems}{Knowledge-Resistant Problems}}
\label{sec:ReallyHard}
Are there problems whose cost is unaffected by \emph{any} checkable knowledge about the input? For undecidable problems such as the Halting Problem \cite{Turing1937Computable} no algorithm exists at all, and restricting them to a decidable property $K$ can make them decidable (e.g.\ programs without loops), so they are not the interesting case. The interesting candidates are the algebraic problems underlying cryptography, such as the
\textit{Discrete Logarithm Problem}\footnote{The Discrete Logarithm Problem (DLP) is formulated as follows: let $G$ be a finite cyclic group of order $q$ with generator $g \in G$. Given an element $y \in G$, find the unique integer $x \in \{0, \dots, q-1\}$ such that $g^x = y$ (denoted $x = \log_g y$). While computing $g^x$ given $x$ (modular exponentiation) is highly efficient, requiring only $O(\log q)$ multiplications, solving the inverse mapping is classically hard. In generic groups (such as elliptic curve groups), the best classical algorithms require exponential time $O(\sqrt{q})$, whereas in multiplicative groups of finite fields $\mathbb{F}_p^*$, index-calculus methods run in sub-exponential time.} \cite{Odlyzko1984DiscreteLI} and \textit{Integer Factorization}\footnote{
	The Integer Factorization Problem (IFP) is formulated as follows: given a composite integer $N$, find its unique prime factors (i.e., write $N = p_1^{a_1} p_2^{a_2} \dots p_k^{a_k}$ where each $p_i$ is a distinct prime). In cryptographic systems (such as RSA), $N$ is chosen as the product of two large, distinct secret primes $p$ and $q$. While verifying a candidate factorization via multiplication is trivial, requiring $O(\log^2 N)$ steps using standard arithmetic, the inverse factorization mapping is classically hard. The fastest known general-purpose classical algorithm is the General Number Field Sieve (GNFS), which operates in sub-exponential time $O(\exp(((\frac{64}{9})^{1/3} + o(1))(\ln N)^{1/3}(\ln \ln N)^{2/3}))$.}.

It is tempting to say that these problems are uniformly hard, i.e.\ that ``almost all instances are hard'' and the cost variation is close to zero. This is not correct, and the reason is instructive. Many checkable properties make factoring easy. If $N$ is prime (checkable in polynomial time \cite{agrawal2004primes}) or a perfect power, the factorization is immediate; if $N$ has a prime factor below a bound polynomial in $\log N$, trial division finds it; if $N=pq$ with $|p-q|$ at most about $N^{1/4}$, Fermat's method succeeds after a constant number of iterations; if $p-1$ is $B$-powersmooth, Pollard's $p-1$ method runs in time polynomial in $B$ and $\log N$ \cite{pollard1974theorems}; the elliptic-curve method finds a prime factor $p$ in time subexponential in $\log p$ \cite{lenstra1987factoring}; and if, in addition to $N$, half of the bits of $p$ are known, $N=pq$ can be factored in polynomial time \cite{coppersmith1997small}. The last example is not a property of $N$ but partial information about the solution, a closely related form of knowledge that collapses the cost just as effectively. Except for the elliptic-curve case, each of the properties of $N$ above can be checked by running the corresponding method for a polynomial number of steps. Moreover, half of all integers are even. The cost variation of any reasonable factoring algorithm is therefore large, not close to zero.

What is believed to be hard is factoring \emph{typical balanced semiprimes}, which form a vanishing fraction of all integers. The right formalization of resistance to knowledge therefore involves the coverage of $K$ with respect to the distribution of interest.

\begin{proposition}[Knowledge resistance of factoring]
\label{prop:factoring}
Let $\mathcal{D}=\{\mathcal{D}_n\}$ be the distribution of $N=pq$ for independent uniformly random $(n/2)$-bit primes $p,q$, and assume the standard factoring assumption: every probabilistic polynomial-time algorithm factors $N\sim\mathcal{D}_n$ with probability negligible in $n$. Let $K$ be a property decidable in polynomial time such that some polynomial-time algorithm factors every $N\models K$. Then the coverage $\mu^{\mathcal{D}}_n(K)$ is negligible.
\end{proposition}

\begin{proof}
Otherwise the algorithm ``check $K$; if it holds, run the polynomial-time factoring algorithm for $K$'' runs in polynomial time and factors $N\sim\mathcal{D}_n$ with probability at least $\mu^{\mathcal{D}}_n(K)$, which is non-negligible, contradicting the assumption.
\end{proof}

The proposition is simple, but it captures exactly the sense in which factoring is \emph{knowledge-resistant}: easy subclasses exist in abundance, yet, under a standard cryptographic assumption, none of them carries non-negligible mass under the distribution that matters. For the discrete logarithm in a fixed group, random self-reducibility \cite{blum1984prob,feigenbaum1991random} makes worst-case and average-case hardness coincide within the group; for factoring no comparable worst-case to average-case reduction is known. Note that random self-reducibility equates the \emph{worst} case with the \emph{average} case, not with the best case: trivial instances (such as $y=g$) still exist, so it does not imply a small cost variation.

Now consider the SUM problem of \S\ref{sec:SU}, which is NSU and easy. NSU-ness and hardness are orthogonal: SUM is NSU and easy; Partition is sensitive and, in the worst case, hard; factoring is sensitive (it has many easy instances) yet, conditionally, knowledge-resistant in the coverage sense. (Both sensitivity claims follow from Proposition~\ref{prop:fastpath}, unless the problems are solvable in linear time.) What separates easy from hard is therefore not the variation $\Delta$ alone, but the joint profile of the restricted costs $W(n,K)$, the residual uncertainty $\Psi(n,K)$, and the coverage of the properties that achieve them.

\subsection{Formal Bridge to Parameterized Complexity}
\label{sec:BridgeToParameterizedComplexity}

To anchor the parametric complexity classes $\UP[K]$ and $\UNP[K]$ within the established taxonomy of theoretical computer science, we formalize their structural relationship to standard \emph{Parameterized Complexity} (the Downey-Fellows framework) \cite{downey1999}. 

In traditional parameterized complexity, a parameterized language is defined as a subset $L \subseteq \Sigma^* \times \mathbb{N}$, where the second coordinate represents a numerical parameter $k$. The class $FPT$ (Fixed-Parameter Tractable) contains all parameterized languages that can be decided by a deterministic Turing machine in time
$
O(f(k) \cdot n^c)
$
where $f$ is an arbitrary computable function of the parameter $k$, $n$ is the input size, and $c$ is a constant independent of $k$.

A single membership $L\in\UP[K]$ is a much weaker statement than fixed-parameter tractability: it concerns one property, whereas FPT is a statement about a whole family of properties, uniformly. The precise relation is as follows. For a parameterization $\kappa:\Sigma^*\to\mathbb{N}$ computable in polynomial time, let $K_k=\{x:\kappa(x)\le k\}$. This is a graded family of properties with $K_0\subseteq K_1\subseteq\cdots$ and $\bigcup_k K_k=\Sigma^*$.

\begin{proposition}[FPT as uniform membership in a graded family]
\label{thm:bridge}
Let $L\subseteq\Sigma^*$ and let $\kappa$ be as above.
\begin{enumerate}
\item $(L,\kappa)\in FPT$ iff there exist a single algorithm $\mathcal{A}$, a computable function $f$ and a constant $c$ such that, for every $k$, $\mathcal{A}$ decides $L$ on all $x\models K_k$ within $f(k)\cdot|x|^{c}$ steps (so that $\mathcal{A}$ witnesses $L\in\UP[K_k]$ for every $k$, with degree independent of $k$).
\item $(L,\kappa)\in XP$ iff the same holds with the bound $|x|^{f(k)}$.
\end{enumerate}
\end{proposition}

\begin{proof}
(1) If $\mathcal{A}$ decides $L$ in time $f(\kappa(x))\,|x|^c$, then on every $x\models K_k$ it runs in time at most $\hat f(k)\,|x|^c$, where $\hat f(k)=\max_{j\le k}f(j)$ is computable. Conversely, if $\mathcal{A}$ satisfies the uniform bound for every $k$, then, since every $x$ satisfies $K_{\kappa(x)}$, $\mathcal{A}$ decides every $x$ correctly within $f(\kappa(x))\,|x|^c$ steps. (2) is identical, with $|x|^{f(k)}$ in place of $f(k)|x|^c$.
\end{proof}

Membership of $L$ in $\UP[K_k]$ for every fixed $k$ is therefore not enough for FPT. For $p$-Clique, parameterized by the size $k$ of the requested clique, brute force gives $L\in\UP[K_k]$ for every $k$ (in time $O(n^{k+2})$), yet $p$-Clique is W[1]-hard and hence not in FPT unless $FPT=W[1]$ \cite{downey1999,cygan2015parameterized}. The uncertainty framework is thus more general than FPT in one respect (it allows arbitrary predicates $K$, not only thresholds of a numerical parameter) and weaker in another (a single statement $L\in\UP[K]$ carries no uniformity).

Where does $\Psi$ enter? If, for the instances satisfying $K$, a determining set and its values can be computed in polynomial time, and the reduced problem maps back to an equivalent instance of $P$ of size $s(\Psi)$ (Corollary~\ref{cor:projection}), then this map behaves like a kernelization of the problem restricted to $K$, with kernel size $s(\Psi)$. If, moreover, it extends to all instances and bounds the parameter of its output, it is a kernel in the usual sense, and since a decidable parameterized problem is in FPT iff it admits a kernel \cite{cygan2015parameterized,fomin2019kernelization}, the problem is fixed-parameter tractable with respect to any parameter that bounds $\Psi$. Without such an extension the conclusion applies only to the restricted problem.

	\begin{remark}[Two regimes]
		\label{rem:bridge-regimes}
		In the \emph{projection regime} the core solver operates on the $\Psi(n,K)$ determining cells, and the relevant parameter is $\Psi$. In the \emph{search-interface regime} the solver reads all $n$ cells but branches only over the active dimension $\alpha(P,K)$ of Definition~\ref{def:activedim}, giving $W(n,K)\le f(\alpha(P,K))\cdot n^{O(1)}$; bounded-treewidth Vertex Cover ($\alpha=\tau$), strong backdoors ($\alpha=k$) and the free variables of $K_f$-instances (\S\ref{sec:ETH}) belong to this regime. In both regimes the statement is of FPT type, in the sense of Proposition~\ref{thm:bridge}, when the bound holds uniformly over the corresponding graded family of properties.
	\end{remark}

\subsection{\texorpdfstring{A Free-Variable Threshold for SAT}{A Free-Variable Threshold for SAT}}
\label{sec:ETH}

A key question is whether $\UP[K]$ and $\UNP[K]$ are distinct for natural properties $K$. Under $P\neq NP$ we have $\UP[\epsilon]\subsetneq\UNP[\epsilon]$, since $\UP[\epsilon]=P$ and $\UNP[\epsilon]=NP$. The interesting case is a property that leaves an unbounded but sub-linear amount of freedom. We consider satisfiability with a partially fixed assignment.

Let $L$ be the language of pairs $(F,\rho)$ such that $F$ is a 3-CNF formula\footnote{
	The 3-Satisfiability (3-SAT) problem is a canonical NP-complete decision problem. Let $V = \{x_1, x_2, \dots, x_n\}$ be a set of Boolean variables. A literal is a variable $x_i$ or its negation $\neg x_i$. A 3-CNF formula $F$ is a conjunction (AND) of clauses, $F = C_1 \land C_2 \land \dots \land C_m$, where each clause $C_i$ is a disjunction (OR) of exactly three literals (e.g., $C_i = l_{i,1} \lor l_{i,2} \lor l_{i,3}$). The problem asks whether there exists a truth assignment $\tau: V \to \{\text{True}, \text{False}\}$ that satisfies $F$, i.e., makes the entire formula evaluate to True.
} over a variable set $V$ with $|V|=n$, $\rho$ is a partial assignment to $V$ given explicitly as part of the input, and $F$ has a satisfying assignment that extends $\rho$. For a function $f$, let $K_f$ be the property ``at most $f(n)$ variables are left unassigned by $\rho$''. The property $K_f$ is checkable in time linear in the input size (count the unassigned variables). Note that $\Psi(\cdot,K_f)$ is \emph{not} $f(n)$: the clauses of $F$ are unconstrained by $K_f$ and must be read, so the residual uncertainty is of the order of the input size. What $K_f$ restricts is the \emph{active dimension} (Definition~\ref{def:activedim}): a solver needs to branch only over the at most $f(n)$ free variables.

We use the \emph{Exponential Time Hypothesis (ETH)} of Impagliazzo, Paturi, and Zane \cite{impagliazzo2001}: 3-SAT over $m$ variables cannot be solved by any deterministic algorithm in time $2^{o(m)}$.

\begin{theorem}[Free-variable threshold]
\label{thm:separation}
\begin{enumerate}
\item If $f(n)=O(\log n)$, then $L\in\UP[K_f]$.
\item Let $\gamma\in(0,1)$ be a constant and $f(n)=n^{\gamma}$. Then $L\in\UNP[K_f]$, and $L\notin\UP[K_f]$ unless $P=NP$.
\item Under ETH, for $f(n)=n^{\gamma}$ with a constant $\gamma\in(0,1)$, no algorithm decides $L$ on all inputs satisfying $K_f$ in time $2^{o(n^{\gamma})}\cdot\mathrm{poly}(|x|)$.
\end{enumerate}
\end{theorem}

\begin{proof}
(1) Enumerate the at most $2^{f(n)}=n^{O(1)}$ extensions of $\rho$ and evaluate $F$ on each of them in linear time.

(2) Membership in $\UNP[K_f]$: guess values for the at most $n^\gamma$ free variables and evaluate $F$, in polynomial time. For the lower bound, let $y$ be an arbitrary 3-CNF formula over $m$ variables. Add $M-m$ dummy variables, where $M=\lceil m^{1/\gamma}\rceil$, which occur in no clause, and let $\rho$ assign all of them to false. The pair $(y,\rho)$ has $n=M$ variables, of which at most $m\le n^{\gamma}$ are free, so it satisfies $K_f$; it belongs to $L$ iff $y$ is satisfiable; and its size is polynomial in $m$, because $\gamma$ is a constant and a 3-CNF formula over $m$ variables has $O(m^3)$ distinct clauses. A polynomial-time algorithm for $L$ on $K_f$-instances would therefore decide 3-SAT in polynomial time.\footnote{A concrete instance of the padding follows. Take $\gamma=\tfrac12$ and $m=9$. The padding appends $m^{1/\gamma}-m=9^{2}-9=72$ dummy variables, all fixed to false, giving $n=81$ variables, of which exactly $m=n^{\gamma}=9$ (the original ones) are free.}

(3) Apply the same padding. On the padded instances $n^{\gamma}=\Theta(m)$ and $|x|=\mathrm{poly}(m)$, so an algorithm running in time $2^{o(n^{\gamma})}\cdot\mathrm{poly}(|x|)$ on them would decide 3-SAT over $m$ variables in time $2^{o(m)}$, contradicting ETH.
\end{proof}

The number of free variables thus has a sharp threshold: logarithmically many keep the problem in $\UP[K_f]$, whereas any polynomial number $n^\gamma$ already makes it as hard as general 3-SAT, up to polynomial padding; under ETH, the brute-force bound $2^{n^\gamma}\cdot\mathrm{poly}(|x|)$ is optimal up to the constant in the exponent. Note that part~(2) needs only $P\neq NP$; ETH is required only for the quantitative bound of part~(3).

\section{\texorpdfstring{Positioning and Discussion}{Positioning and Discussion}}
\label{sec:Originality}

We first position the framework with respect to other refinements of worst-case analysis (Section~\ref{sec:superiority}), then compare the parameterized and the uncertainty views on the Partition problem (Section~\ref{app:Partition}), and finally examine further NP-hard problems (Section~\ref{sec:MoreCollapse}) and tractable tasks (Section~\ref{sec:tractable}).

\subsection{\texorpdfstring{Relation to Other Refinements of Worst-Case Analysis}{Relation to Other Refinements of Worst-Case Analysis}}
\label{sec:superiority}

\noindent 
{\em How does this differ from standard Fixed-Parameter Tractability (FPT)?}
As described in \S \ref{sec:BridgeToParameterizedComplexity},
Fixed-Parameter Tractability (FPT) is a framework in computational complexity that studies whether a hard problem becomes efficiently solvable when some aspect of the input (the parameter) is small.
It is the central concept in parameterized complexity theory.
The basic idea is that
many NP-hard problems are difficult because the input size $n$ is large. However, in practice, some property of the input, such as the size of the desired solution, the treewidth of a graph, or the number of variables, may be small.
Parameterized complexity separates these two quantities:
(a) $n$: the overall input size,
(b) $k$: a parameter that is expected to be small.
Instead of asking whether a problem can be solved in polynomial time in $n$, we ask whether it can be solved in
$f(k) \cdot n^{O(1)}$,
where
$f(k)$ is any computable function depending only on $k$,
$n^{O(1)}$  is a polynomial in the input size.
If such an algorithm exists, the problem is fixed-parameter tractable (FPT).
This means that in FPT, we already parameterize complexity using a parameter $k$. Why do we need $\UP[K]$?
Standard FPT parameterizes the numerical properties of the input (e.g., vertex cover size $k$).
$\UP[K]$ parameterizes the logical predicate $K$ and its information content $\Psi(n,K)$. 
Numerical parameterizations are recovered as graded families of threshold predicates (Proposition~\ref{thm:bridge}). Below we discuss in what sense the resulting bounds differ from classical and pseudo-polynomial ones, and how the framework relates to average-case analysis.

\ \\
\noindent 
\textbf{A. Removing pseudo-polynomial dependencies.} Classical dynamic programming for numerical NP-complete problems yields bounds such as $O(n\cdot B)$, where $B$ is the sum of the values, i.e.\ exponential in their bit length $b$. Under $K_{\text{allEqual}}$ the guarded algorithm uses $O(n)$ word operations, independently of $b$ (one comparison per cell). This is an instance of a general and well-known phenomenon (restrictions can remove pseudo-polynomial dependencies), not a new type of bound; what the framework adds is that the restriction and the cost of checking it are stated explicitly.

\ \\
\noindent 
\textbf{B. Information lower bounds.} In the cell-probe model, the task-relative metric gives a lower bound on the promise cost of \emph{every} algorithm, while Theorem~\ref{thm:uncertainty} gives an upper bound for the specific algorithm $\mathcal{A}_K$ run under the promise. By Remark~\ref{rem:recon-vs-task} and Theorem~\ref{thm:uncertainty}, for every algorithm $\mathcal{A}$ correct under the promise $K$,
\[
\Psi_{\mathrm{task}}(n,K,P)\;\le\; W^{\mathrm{prom}}_{\mathcal{A}}(n,K),
\qquad
W^{\mathrm{prom}}_{\mathcal{A}_K}(n,K)\;\le\; C_{\mathrm{find}}(n,K)+O\bigl(\Psi(n,K)\bigr)+T_{\mathcal{B}}\bigl(n,\Psi(n,K)\bigr).
\]
The two sides can be far apart: for Partition with $K=\epsilon$ they are $n$ and exponential, respectively. They match up to constants in the constant-uncertainty cases of \S\ref{sec:ThePartitionProblem}: under $K_{\mathrm{allEqual}}$ and under $K_{\mathrm{arith}}$ the decision depends only on $n$, so $\Psi_{\mathrm{task}}=0$ and the promise cost is $O(1)$. The lower bound must use $\Psi_{\mathrm{task}}$ rather than $\Psi$: for lookup in a sorted array $\Psi=n$, yet only $O(\log n)$ cells need to be read (\S\ref{sec:tractable}).

\ \\
\noindent
{\bf C. How does this differ from average-case complexity?}
Average-case complexity \cite{levin1986avg,ben-david1992avg,bogdanov2006avg}
measures, for a distribution $\mathcal{D}=\{\mathcal{D}_n\}$ over
instances, the expected cost
$\mathrm{avgCost}_{\mathcal{D}}(n)=\mathbb{E}_{x\sim\mathcal{D}_n}
[T_{\mathcal{A}}(x)]$. It answers: \emph{is the typical instance easy?}
The present framework answers a different question:
\emph{is there a checkable structural property $K$ such that every
instance satisfying $K$ is easy?} The two views are complementary:
\begin{itemize}
	\item \emph{Quantitative bridge.} If the support of $\mathcal{D}_n$ is
	contained in $\{x\models K\}$, then
	$\mathrm{avgCost}_{\mathcal{D}}(n)\le W_{\mathcal{A}}(n,K)$;
	in particular, $P\in\UP[K]$ implies that $P$ is average-case polynomial
	for \emph{every} distribution concentrated on $\{x\models K\}$. The
	converse fails in general: a problem may be easy on random instances without
	admitting any \emph{recognizable} property $K$ that certifies easiness. The closest average-case notions are \emph{errorless heuristics} \cite{bogdanov2006avg} and \emph{generic-case} algorithms \cite{kapovich2003generic}, which never err but may answer ``don't know''. Such an algorithm implicitly defines the property ``the algorithm answers'', but this property need not be cheaper to check than running the algorithm. The $\UP[K]$ view makes the guard an explicit, named predicate with its own cost $C_{\mathrm{verify}}$; the average-case view adds \emph{typicality} across all instances.
	\item \emph{Worst-case vs.\ average-case for hard problems.} For
	problems with random self-reducibility
	\cite{blum1984prob,feigenbaum1991random}, notably the discrete logarithm in a fixed group, worst-case and
	average-case hardness coincide; this is why cryptography can rely on average-case assumptions for them, and it is consistent with knowledge resistance in the coverage sense of Proposition~\ref{prop:factoring}. It does not imply a small cost variation, since trivial instances still exist. At the opposite end, number
	partitioning is \emph{easy on average} in the regime where the number of items is large compared with the bit length of the values
	(Mertens calls it ``the easiest hard problem'' \cite{mertens2006easiest}; see also \cite{borgs2001phase} for the underlying phase transition), while its
	worst case is exponential. The easiness of random instances is established statistically, without a per-instance certificate; certifiable structure and typicality are thus complementary refinements of worst-case complexity.
\end{itemize}

	\paragraph{Is the metric $\Psi(n,K)$ computable?}
		    When the cell alphabet $\Sigma$ is finite, $I_n$ is finite and $\Psi(n,K)$ is a max--min over finitely many instances and cell subsets, so for every \emph{decidable} property $K$ it is computable by exhaustive search (though not necessarily efficiently); for undecidable $K$ it need not be. With an infinite alphabet (e.g.\ unbounded integers), checking whether $Q$ determines $x$ quantifies over infinitely many inputs and need not be decidable in general; for the concrete properties of this paper $\Psi$ is determined by direct argument. This is a different matter from the
		Kolmogorov complexity $\kappa(x\mid K)$, which is itself uncomputable even
		though Lemma~\ref{lem:compression} bounds $\kappa(x\mid K,n)$ above by
		$\Psi(n,K)\log|\Sigma|+O(1)$; the uncomputability of $\kappa$ does not by
		itself transfer to $\Psi$. For practical use we restrict $K$ to tractable
		predicates for which $\Psi(n,K)$ is efficiently computable (e.g., in
		polynomial time).

\ \\
\noindent 
{\bf Is the Uncertainty Reduction Theorem model-dependent?}
The theorem charges $O(1)$ per cell read, which presupposes random access; this is why we adopted the word RAM in \S\ref{sec:SU}. On a single-tape Turing machine, reaching the cells of $Q$ incurs traversal costs; if $Q$ is sorted, a single left-to-right sweep reads all of it in $O(n)$ steps, so the reading term becomes $O(n)$ (and the prefix variant $\Psi_{\mathrm{pref}}$ avoids the issue altogether). The locating charge $C_{\mathrm{find}}(n,K)$, however, is not an artifact of the machine model: it is intrinsic to the existential definition of $\Psi(n,K)$, and it is exactly what Theorem~\ref{thm:uncertainty} makes explicit.

\ \\
{\bf Relation to non-deterministic computation trees (speculative).}
A property $K$ can also be viewed as pruning the computation tree of a non-deterministic machine: under $K_f$ of \S\ref{sec:ETH}, for instance, the relevant branching is over $2^{f(n)}$ assignments instead of $2^{n}$. Making this precise beyond such examples, e.g.\ as a ``branching entropy'' of the computation tree restricted to $\{x\models K\}$, is left open.

\subsection{Partition Problem: Fixed-Parameter Tractability vs Sensitivity to Uncertainty}
\label{app:Partition}


\subsubsection{Analyzing the Partition Problem through FPT}

Below we describe the complexity of the partition problem (of a set of positive integers $a_1,\ldots,a_n$),
according to various parameters, in a FPT style.
\begin{enumerate}
	\item  Parameter {\em Target Sum.}

	Let $S = \frac{1}{2}\sum_{i=1}^n a_i$.
	The classical dynamic programming algorithm solves Partition in time
	$O(nS)$. This running time can be written as $f(S) \cdot n^{O(1)}$,
	where $f(S) = S$. Since $f(S)$ is a computable function of the
	parameter $S$ alone, Partition is \emph{fixed-parameter tractable
		(FPT)} with respect to the target sum $S$.

	\item Parameter {\em Maximum Element.}
	Let
	$
	M=\max_i a_i.
	$
	Since
	$
	\sum_i a_i \le nM,
	$
	the dynamic programming algorithm runs in
	$
	O(n^2M).
	$
	Equivalently,
	$
	O(Mn^2)=f(M)n^2,
	$
	where $f(M)=M$.
	Hence, Partition is also \textit{FPT} when parameterized by the maximum input value.
	
	\item Parameter {\em Number of Items.}
	If the parameter is simply the number of integers $n$, exhaustive search examines all subsets in
	$
	O(2^n)
	$
	time.
	Since this can be written as
	$
	2^n=f(n),
	$
	Partition is trivially FPT with respect to $n$. However, this parameterization is generally considered uninteresting because the parameter is essentially the entire input size.
	
	\item Parameter {\em Number of Distinct Values.}
	Suppose the input contains only $k$ distinct integer values.
	In this case, several algorithms based on compressed dynamic programming and integer programming techniques (e.g.\ integer programming in fixed dimension \cite{lenstra1983integer}) show that Partition (and the related Subset Sum problem) is fixed-parameter tractable with respect to $k$.
	
	\item Parameter {\em Allowed Imbalance.}
	A related optimization problem asks whether the integers can be partitioned into two subsets whose sums differ by at most $k$.
	The parameterized complexity of this variant depends on the precise formulation, but several versions admit FPT algorithms and parameterized approximation schemes.

\end{enumerate}


Also note that the Partition is a special case of the \emph{Subset Sum} problem,
meaning that if we set $
T=\sum_{i=1}^{n} a_i
$,
then Partition is equivalent to asking whether there exists a subset whose sum is
$
\frac{T}{2}.
$
Consequently, many parameterized algorithms and complexity results for {\em Subset Sum} immediately apply to Partition.

\ \\

\noindent
{\bf Summary.}
A summary of the partition problem from a FPT perspective,  is given in Table \ref{tab:PartFPT}.
Note that the Partition problem is considerably more tractable than many other NP-complete problems from a parameterized perspective. 
In particular, parameterizations based on numeric quantities such as the target sum or the maximum input value lead to fixed-parameter tractable algorithms via classical pseudo-polynomial dynamic programming. 
This contrasts with strongly NP-complete problems, for which such pseudo-polynomial and FPT algorithms are generally not expected to exist.

\begin{table}[htbp]
	\begin{center}
		\begin{tabular}{|l|c|l|}
			\hline
			\textbf{Parameter} & \textbf{FPT?} & \textbf{Reason} \\
			\hline
			Target sum $S$ & Yes & Dynamic programming in $O(nS)$ \\
			\hline
			Maximum value $M$ & Yes & Since $S\le nM$, DP runs in $O(n^2M)$ \\
			\hline
			Number of distinct values & Yes & Compressed DP / integer programming \\
			\hline
			Number of items $n$ & Yes (trivial) & Exhaustive search in $O(2^n)$ \\
			\hline
			None (polynomial time) & No (unless $P=NP$) & Partition is NP-complete \\
			\hline
		\end{tabular}
		\caption{Partition as FPT problem}
		\label{tab:PartFPT}
	\end{center}
\end{table}


\subsubsection{Analyzing the Partition Problem through Sensitivity to Uncertainty}

The relationship between $\UP[K]$ and standard parameterized complexity is one of \emph{complementarity}. By Proposition~\ref{thm:bridge}, FPT corresponds to uniform membership in a graded family $\UP[K_k]$, and, when the residual information is efficiently locatable and maps back to a small instance of the same problem, the reduction of Theorem~\ref{thm:uncertainty} is analogous to a kernelization of the restricted problem. The parameterizations of Table~\ref{tab:PartFPT} are of this graded kind; the properties of \S\ref{sec:ThePartitionProblem} (all-equal, arithmetic sequence) are individual predicates, and their natural parameterization (e.g.\ by the number of distinct values, which is $1$ under all-equal) already appears in the table. What the uncertainty view contributes on Partition is therefore not new tractable cases but a different accounting:

\begin{enumerate}
	\item \textit{Predicates as parameters.} Standard FPT parameterizes a problem by a numeric feature of the input (treewidth, solution size, maximum value); the present framework starts from a logical predicate $K$ and measures the residual information $\Psi(n,K)$ that it leaves. Numeric parameterizations are the special case of threshold predicates $K_k$.
	\item \textit{Explicit accounting of verification and find costs.} Parameterized algorithms often assume that a structure witnessing the parameter is given (e.g.\ a tree decomposition), although computing it is also studied (e.g.\ by Bodlaender's algorithm \cite{bodlaender1996linear} or Korhonen's approximation \cite{korhonen2021single}). The framework makes the cost of checking $K$, $C_{\mathrm{verify}}$, and the cost of locating the residual information, $C_{\mathrm{find}}$, mandatory terms of every bound.
	\item \textit{Promise versus guarded cost.} In the constant-uncertainty regime ($\Psi(n,K)=O(1)$), the promise cost can be $O(1)$, i.e.\ below the $f(k)\cdot n^{c}$ form, because the input need not even be read. Once the guard is charged, however, the guarded cost $C_{\mathrm{verify}}(n,K)+C_{\mathrm{find}}(n,K)+O(1)$ is again of the usual form.
\end{enumerate}

We therefore do \emph{not} claim that $\UP[K]$ is superior to FPT as a framework. An open direction is a hardness hierarchy for \emph{locating} tractable properties $K$ (a $W$-analogue for the find cost), which would tell us which instance classes are not merely unparameterized but genuinely hard to exploit.

\subsection{Collapse to Polynomial: Further Case Studies}
\label{sec:MoreCollapse}

We now examine three well-known $\mathrm{NP}$-hard problems where structural
knowledge $K$ provably collapses the exponential worst case to polynomial
time, and where the collapse is made transparent by the metrics
$(W(n,K), \Psi(n,K))$ of this framework.

\subsubsection{Boolean Satisfiability with Strong Backdoor Sets}

The \emph{Boolean Satisfiability} problem (SAT) asks: given a propositional
formula $F$, typically in conjunctive normal form, i.e.\ a conjunction of
clauses, each clause being a disjunction of literals, does there exist a
truth assignment to its variables that makes $F$ evaluate to true? SAT was the
first problem proved $\mathrm{NP}$-complete (Cook--Levin), and no algorithm
substantially faster than exhaustive search over all $2^{n}$ assignments is
believed to exist (the Strong Exponential Time Hypothesis, SETH). 
Yet, 
its worst-case complexity collapses once the formula
carries structural information, in particular, a small set of variables whose
fixing reduces the rest of the formula to a tractable class.

Let $\mathcal{C}$ be a tractable class of formulas (e.g.\ 2-SAT or Horn).
A set $B$ of variables of a CNF formula $F$ is a \emph{strong backdoor set}
of size $k=|B|$ with respect to $\mathcal{C}$ if for \emph{every} partial
assignment $\tau : B \to \{0,1\}$ (there are $2^k$ of them) the reduced
formula $F[\tau]$ belongs to $\mathcal{C}$ \cite{williams2003backdoors,gaspers2021backdoors}.
The mapping to our perspective follows:

\begin{itemize}
	\item \textit{Property $K$:} $F$ is a 3-CNF formula that has a strong backdoor set of size at most $k$ with respect to 2-CNF.
	\item \textit{Checking $K$:} a set $B$ is a strong 2-CNF backdoor of a 3-CNF formula iff every clause with three distinct variables contains a variable of $B$. Indeed, an assignment to $B$ that falsifies all $B$-literals of a clause leaves its other literals in place, so every clause must have at most two variables outside $B$. Finding such a $B$ of size at most $k$ is therefore a 3-Hitting Set problem, solvable by bounded branching in time $O(3^k\cdot|F|)$ \cite{gaspers2021backdoors}. For $k=O(\log n)$ this check is polynomial, and here $C_{\mathrm{verify}}(K)$ includes finding $B$.
	\item \textit{Metrics:} the whole formula must still be read, so
	$\Psi(n,K)=\Theta(|F|)$; what collapses is the \emph{active dimension}
	$\alpha=k$ (Def.~\ref{def:activedim}), the number of backdoor variables the
	solver branches over (search-interface regime,
	Remark~\ref{rem:bridge-regimes}).
	\item \textit{Cost collapse:} enumerate the $2^k$ assignments of $B$ and run a
	linear-time 2-SAT solver on each reduced formula; together with the check, the guarded cost is $O(3^{k}\cdot|F|)$. If $k=O(\log n)$, this is polynomial.
\end{itemize}

For example, let $F=(x_1\lor x_2)\land(\neg x_2\lor x_3)\land(x_1\lor\neg x_3\lor x_4)$.
The only clause with three literals is $C_3=(x_1\lor\neg x_3\lor x_4)$, so $B=\{x_4\}$ is a strong backdoor set with respect to 2-CNF: if $x_4\mapsto\top$, $C_3$ is satisfied and disappears; if $x_4\mapsto\bot$, $C_3$ becomes the 2-clause $(x_1\lor\neg x_3)$. In both cases the reduced formula $F[\tau]$ is a 2-CNF and is decidable in linear time, so $F$ is decided by solving the $2^{|B|}=2$ reduced formulas. The answer still depends on the whole formula; what the backdoor bounds is the \emph{branching} of the solver, $\alpha=|B|=1$.

\noindent
Consequently $\mathrm{SAT}\in\UP[K]$ for $k=O(\log n)$, although 3-SAT is
$\mathrm{NP}$-complete at $K=\epsilon$: the exponential penalty has moved from the number of variables to the backdoor size $k$. For general CNF formulas and other base classes, detecting small strong backdoors can be harder, and the complexity of backdoor detection depends on the base class \cite{gaspers2021backdoors}; there, $K$ is best treated as a certified promise, with $B$ supplied as part of the input.

\subsubsection{Vertex Cover on Graphs of Bounded Treewidth}

The \emph{Vertex Cover} problem asks: given an undirected graph $G=(V,E)$ with
$n$ vertices, find the smallest set of vertices $C\subseteq V$ such that every
edge $e\in E$ is incident to at least one vertex of $C$ (i.e.\ every edge is
\emph{covered}). The decision version asks whether such a set of size at most
$k$ exists. Vertex Cover is one of the classic $\mathrm{NP}$-complete problems,
with a naive exhaustive search over all $2^n$ vertex subsets. 
Yet, 
as with SAT,
its worst-case complexity collapses once the graph carries structural
information, in particular, a bounded \emph{treewidth} $\tau$.
In particular, suppose the input graph $G$ has treewidth $\tau$, i.e.\ a nice tree decomposition
of width $\tau$ is available \cite{robertson1986treemwidth,bodlaender1996linear,arnborg1989partial}.
The mapping to our perspective follows:

\begin{itemize}
	\item \textit{Property $K$:} $\mathrm{tw}(G)\le \tau$ (i.e. decomposition given). If the decomposition is not given, an optimal one can be computed in time $2^{O(\tau^3)}\cdot n$ \cite{bodlaender1996linear}, and one of width at most $2\tau+1$ in time $2^{O(\tau)}\cdot n$ \cite{korhonen2021single}. The latter suffices for the dynamic program (with $2^{2\tau+1}$ in place of $2^{\tau}$), so for $\tau=O(\log n)$ the find cost and the whole computation remain polynomial. This is a clear instance of a find cost $C_{\mathrm{find}}$ whose size depends on the algorithm used to locate the structure.
	\item \textit{Uncertainty metric:} the whole graph must still be read
	(i.e. $\Psi(n,K)=n$); what collapses is the \emph{active dimension} of the
	search: the dynamic program over the decomposition passes only
	$O(\tau)$-sized interface states between adjacent bags.
	\item \textit{Cost collapse:} the standard dynamic programming algorithm over the decomposition runs in
	$
	W^{\mathrm{prom}}_{\mathrm{VC}}(n,K)\ \le\ O(2^{\tau}\cdot\tau^{O(1)}\cdot n),
	$
	polynomial for $\tau=O(1)$, and $n^{O(1)}$ for $\tau=O(\log n)$.
\end{itemize}

This case is interesting because the mechanism differs from the previous ones:
the collapse is not obtained by reading fewer input cells but by
\emph{parameterizing the search interface}: this is the search-interface
regime of the Bridge to FPT (Remark~\ref{rem:bridge-regimes}), with active
dimension $\alpha=\tau$ (Def.~\ref{def:activedim}).
It thus illustrates how the
framework absorbs classical FPT results as instances of knowledge-driven
tractability.

\subsubsection{\texorpdfstring{The Traveling Salesperson Problem with Monge-Type Matrices}{The Traveling Salesperson Problem with Monge-Type Matrices}}
\label{sec:TSPwithMonge}

The \emph{Traveling Salesperson Problem} (TSP) asks for the shortest Hamiltonian
cycle visiting $n$ cities given their pairwise distances. In full generality it
is $\mathrm{NP}$-hard, and the canonical dynamic program
(Bellman--Held--Karp) runs in $O(2^{n}n^{2})$ time, exponential even for
moderate $n$. Yet, as with SAT and Vertex Cover, this pessimistic bound
collapses the moment the distance matrix carries algebraic structure,
in particular, when it satisfies certain ordering inequalities of Monge type.

Three classical well-solvable cases are relevant \cite{burkard1998wellsolvable,burkard1996perspectives}. A matrix $D$ is a \emph{Monge} matrix if $D[i,j]+D[k,l]\le D[i,l]+D[k,j]$ for all $i<k$ and $j<l$; a symmetric Monge matrix is also called a \emph{Supnick} matrix. A symmetric matrix $D$ is a \emph{Kalmanson} matrix if $D[i,j]+D[k,l]\le D[i,k]+D[j,l]$ and $D[i,l]+D[j,k]\le D[i,k]+D[j,l]$ for all $i<j<k<l$. \emph{Demidenko} matrices form a broader asymmetric class that contains the Monge matrices.
\begin{itemize}
	\item \textit{Supnick (symmetric Monge):} the tour $1\to3\to5\to\cdots\to n\to\cdots\to6\to4\to2\to1$ (odd cities in increasing order, then even cities in decreasing order) is optimal \cite{supnick1957extreme}.
	\item \textit{Kalmanson:} the identity tour $1\to 2\to \dots \to n\to 1$ is optimal \cite{kalmanson1975edgeconvex}.
	\item \textit{Demidenko (in particular, asymmetric Monge):} some pyramidal tour is optimal, and an optimal pyramidal tour can be computed in $O(n^2)$ time by dynamic programming \cite{burkard1998wellsolvable}; the identity tour is not optimal in general.
\end{itemize}
The mapping to our perspective follows:
\begin{itemize}
	\item \textit{Property $K$:} $D$ is a Supnick (respectively, Kalmanson) matrix.
	\item \textit{Checking $K$:} a matrix is Monge iff all its adjacent $2\times2$ submatrices satisfy the Monge inequality, so $K_{\mathrm{Supnick}}$ can be checked with $O(n^2)$ comparisons, i.e.\ in time \emph{linear} in the input size of $n^2$ cells \cite{burkard1996perspectives}; checking all Kalmanson inequalities naively takes $O(n^4)$ time.
	\item \textit{Uncertainty metrics:} $K$ leaves the matrix largely undetermined, so $\Psi(n^2,K)$ is not small. However, for the task ``output an optimal tour'' $K$ fixes the answer completely, so $\Psi_{\mathrm{task}}=0$: no cell needs to be read. Computing the length of that tour, or answering the decision version ``is there a tour of length at most $B$?'', requires reading the $n$ entries along the tour, so $\Psi_{\mathrm{task}}\le n$.
	\item \textit{Cost collapse:} the promise cost is $O(n)$ (to compute the length of the fixed tour), versus $O(2^n n^2)$ for general TSP via Bellman--Held--Karp; for Demidenko matrices it is $O(n^2)$.
\end{itemize}

\noindent
Consequently, $\mathrm{TSP} \in \UP[K]$ for these properties, despite the general problem being $\mathrm{NP}$-hard, and the algorithm that sums the $n$ entries along the fixed tour is NSU on $\{x\models K\}$. This is a third mechanism of tractability: the property fixes the optimal \emph{solution} rather than the input.

\ \\
\smallskip
\noindent\textbf{Summary.} Together with the Partition case, these examples exhibit three distinct mechanisms of tractability, which a single worst-case bound aggregates: \\
$\bullet$ \emph{reduction of residual uncertainty} (constant $\Psi$ or $\Psi_{\mathrm{task}}$, as in Partition under all-equal), \\
$\bullet$ \emph{FPT-style interface parameterization}, which is polynomial when the active dimension is logarithmic ($\alpha=O(\log n)$ for SAT backdoors, $\tau=O(\log n)$ for bounded treewidth, with the decomposition given or computed approximately \cite{korhonen2021single}), and \\
$\bullet$ \emph{algebraic structure that fixes the optimal solution} (TSP on Supnick, Kalmanson and Demidenko matrices), where $\Psi_{\mathrm{task}}$ is small although $\Psi$ is not.

\subsection{The Framework in Tractable Tasks}
\label{sec:tractable}

Note that the uncertainty analysis is not restricted to NP-hard problems: 
it applies to
any computational task, including polynomial-time ones. Consider 
{\em membership lookup} in an array $A[1..n]$: 
without structural knowledge ($K=\epsilon$) the
worst case is a linear scan, so $\Delta_{\mathrm{Lookup}}(n,\epsilon)=O(n)$.
If the array is known to be {\em sorted} ($K_{\mathrm{sorted}}$), 
{\em binary search} answers the query by reading only $O(\log n)$ cells, hence
$\Delta_{\mathrm{Lookup}}(n,K_{\mathrm{sorted}})=O(\log n)$,
meaning that  the cost
variation collapses from linear to logarithmic.
Note that this collapse is
\emph{not} a reduction of input uncertainty, reconstructing the whole
array still requires reading all $n$ cells, so
$\Psi(n,K_{\mathrm{sorted}})=n$, but a reduction of \emph{task}
uncertainty: the query can be resolved by reading only the cells that the
search probes. The task-relative variant, let's denote it by  $\Psi_{\mathrm{task}}(n,K)$
(as introduced in 
\S  \ref{sec:UncertaintyReductionTheorem}, Remark \ref{rem:recon-vs-task}), i.e.
the
minimum number of cells that must be read to resolve the task (not to
reconstruct the input), makes this precise:
$\Psi_{\mathrm{task}}(n,K_{\mathrm{sorted}})=O(\log n)$, and the same
uncertainty-reduction reasoning applies verbatim with $\Psi_{\mathrm{task}}$
in place of $\Psi$. This shows that Sensitivity to Uncertainty is a
continuous spectrum: exponential collapse for NP-hard problems with strong
structure, logarithmic collapse for tractable problems (like lookup on
sorted data), 
and, for knowledge-resistant problems such as Integer Factorization on balanced semiprimes, no collapse on any property of non-negligible coverage (conditionally on the factoring assumption, Proposition~\ref{prop:factoring}).

\section{Applications in Software Engineering, Agentic AI and Quantum Computing}
\label{sec:SE}

The accounting of the previous sections also has practical uses. Software architects, compiler designers, and database engineers routinely bypass worst-case bottlenecks based on instance-level observations. In this section we show how the framework turns such observations into explicit, checkable contracts, and we relate them to established techniques such as algorithm selection and sequential testing.

\subsection{Predicate-Based Algorithm Dispatching}

In software architecture, developers frequently encounter tasks that are globally $NP$-hard (such as bin-packing for cloud resource allocation, register allocation in compilers, or subset-sum checks in financial ledgers). Theorem~\ref{thm:uncertainty} suggests a design pattern that we call \emph{Predicate-Based Algorithm Dispatching}. It is a special case of Rice's algorithm selection problem \cite{rice1976algorithm}, in which the selector is a set of checkable predicates rather than, for example, a learned model as in portfolio solvers \cite{xu2008satzilla}.

Rather than routing all arbitrary inputs of size $n$ to an expensive, general-case exponential solver, the system can be  designed as a two-tier architectural mediator:
\begin{enumerate}
	\item \textit{The Guard Predicate ($C_{\text{verify}}(K)$):} A lightweight, deterministic polynomial-time checker that scans the incoming instance $x$ to verify if it satisfies a pre-specified structural property $K$.
	\item \textit{The Dispatcher:} If the check succeeds ($x \models K$), the dispatcher routes the input to a specialized, fast solver that exploits $K$. If the check fails, the dispatcher falls back to the generic $NP$-hard engine.
\end{enumerate}

By Theorem~\ref{thm:uncertainty}, the running time of this execution path on an input $x\models K$ is bounded by
\[
T_{\mathcal{A}_K}(x)\;\le\;
C_{\mathrm{verify}}(n,K)+C_{\mathrm{find}}(n,K)+O\bigl(\Psi(n,K)\bigr)+T_{\mathcal{B}}\bigl(n,\Psi(n,K)\bigr),
\]
where $\mathcal{B}$ solves the reduced problem on the $\Psi(n,K)$ determining cells. The dispatch is beneficial when this bound is smaller than the cost of the general solver; on inputs that fail the guard, the verification cost is pure overhead, which is why the order of guards matters (\S\ref{sec:GuardOrdering}). For the structural predicates in this paper's case studies (all-equal, 2-CNF, Horn, tree, bipartiteness), $C_{\mathrm{verify}}(K)$ is $O(n)$ (or $O(|V|+|E|)$) and $C_{\mathrm{find}}(n,K)$ is $O(1)$ or immaterial, so the latency bound of \S\ref{sec:SLA} is linear plus the solver term. For predicates with super-linear verification (e.g.\ recognition of perfect graphs, $O(|V|^9)$), the same architecture still yields a polynomial bound, but the verification cost becomes the dominant term.

\subsection{Deterministic SLAs via Design-by-Contract in $\UP[K]$}
\label{sec:SLA}

In modern distributed microservice architectures, API gateways must guarantee strict, deterministic \emph{Service Level Agreements} (SLAs) for maximum response latency (e.g., ``this endpoint will return a result in under $50\text{ms}$''). Under classical worst-case complexity theory, if a microservice relies on a backend engine that solves an $NP$-complete problem, 
guaranteeing a reliable SLA in polynomial time is impossible without
assumptions on the input (unless $P=NP$).

By adopting the classes $\UP[K]$, software engineers can obtain explicit worst-case latency bounds using the \emph{Design-by-Contract} paradigm. Rather than defining a weak API contract with an unstructured parameter (such as \texttt{List<int> input}), the gateway enforces a structural precondition corresponding to a property $K$ (such as \texttt{AllEqualList<int> input} for Partition, \texttt{BipartiteGraph input} for Vertex Cover, or \texttt{TwoCNF input} for SAT):
\begin{enumerate}
	\item The contract names $K$, and the gateway checks it at cost $C_{\mathrm{verify}}(n,K)$ (or trusts a certificate provided upstream).
	\item Because the backend problem is in $\UP[K]$ through an explicit algorithm, Theorem~\ref{thm:uncertainty} yields an explicit polynomial bound on the number of operations as a function of $n$.
	\item Together with a size limit on admissible requests and a calibration of the cost per operation on the target hardware, this gives an \emph{algorithmic} latency bound. It is deterministic in the sense that it holds for every admissible input; actual latencies also depend on hardware, load and constant factors, which is why measured SLA checks can still fail (Appendix~\ref{sec:ImplResults}).
\end{enumerate}


\subsubsection{Proof of Concept Implementation}

An implementation of the Design-by-Contract
for the Partition problem, 
as well as for all problems 
discussed in \S \ref{sec:MoreCollapse} (SAT, Vertex Cover, TSP),
is available  at 
\url{https://github.com/YannisTzitzikas/PNPu}.
For example, the output of the execution of the partition problem, 
is shown below.

\begin{Verbatim}[frame=lines,breaklines=true,fontsize=\scriptsize]
================================================================================
DESIGN-BY-CONTRACT & PARAMETRIC COMPLEXITY (P_u[K]) DEMONSTRATION
================================================================================

--> Test 1: Processing input of size N = 10,000,000 satisfying Property K_2 (All Equal)...
- Partition Result   : TRUE (Equal Partition Exists)
- Matched Property K : K_ALL_EQUAL (s_i = k)
- Guard Verify Time  : 2.140 ms (C_verify(K) linear scan)
- Solver Exec Time   : 0.000300 ms (P_u[K] O(1) Collapse)
- Total Time         : 2.140 ms
- SLA Status         : PASSED [SLA MET]

--> Test 2: Processing input of size N = 10,000,000 satisfying Property K_3 (Sequential {1..N})...
- Partition Result   : TRUE (Equal Partition Exists)
- Matched Property K : K_ARITHMETIC_SEQ ({1..n})
- Guard Verify Time  : 2.315 ms (C_verify(K) linear scan)
- Solver Exec Time   : 0.000200 ms (P_u[K] O(1) Collapse)
- Total Time         : 2.315 ms
- SLA Status         : PASSED [SLA MET]

--> Test 3: Testing small unstructured input (N = 30, K = \epsilon)...
- Partition Result   : TRUE (Equal Partition Exists)
- Matched Property K : K_UNSTRUCTURED (Fallback DP)
- Guard Verify Time  : 0.002 ms (C_verify(K) linear scan)
- Solver Exec Time   : 0.018000 ms (P_u[K] O(1) Collapse)
- Total Time         : 0.020 ms
- SLA Status         : PASSED [SLA MET]

--> Test 4: Attempting large unstructured input (N = 1,000, K = \epsilon)...
[CONTRACT REJECTION TRIGGERED]
Exception Message: Design-by-Contract Violation: Input instance lacks structural property K (K=\epsilon) and size N=1000 exceeds deterministic SLA execution bounds.
SLA Status       : PROTECTED (System prevented unbounded CPU execution!)

================================================================================
SUMMARY OF PRACTICAL VALUE:
1. Verifying property K in O(n) collapses NP-hard solving to O(1) in P_u[K].
2. Execution time for 10 Million elements took < 3ms, easily meeting <5ms SLAs.
3. Small unstructured inputs meet the SLA.
4. Unstructured inputs are rejected upfront, guaranteeing deterministic latency.
================================================================================
\end{Verbatim}

The key software engineering takeaways illustrated by the code are:
\begin{enumerate}
	\item \emph{Explicit bounds under a contract.} In Test~1, with $N=10{,}000{,}000$ elements, exhaustive search would be hopeless, and even the pseudo-polynomial dynamic program would need time proportional to $N$ times the sum of the values; under the all-equal contract the total execution time is $2.14$ ms (meeting the $5$ ms SLA threshold).
	\item \emph{Direct mapping to Theorem~\ref{thm:uncertainty}.} In Test~1 the guard $C_{\mathrm{verify}}(K)$ takes $2.14$ ms (the linear scan), whereas the reduced problem (a parity test) takes $0.0003$ ms: as predicted, the guard dominates.
	\item \emph{Protection against pathological latency.} In Test~4 an input without a recognized property is rejected upfront instead of being handed to an exponential solver. This is a policy decision based on input size, not a complexity-theoretic fact: depending on the magnitude of its values, an instance with $N=1{,}000$ may well be solvable quickly by dynamic programming.
\end{enumerate}

\noindent
Screenshots of the execution of the rest programs are available in 
Appendix \ref{sec:ImplResults}.

\subsubsection{Optimal Guard Ordering for Multiple Structural Properties}
\label{sec:GuardOrdering}

The framework so far assumed a single structural property $K$ whose guard
$C_{\mathrm{verify}}(K)$ is checked against every incoming instance. In
practice a dispatcher may test a {\em collection} of properties
$K_1,\dots,K_m$, each yielding a polynomial collapse
$P \in \UP[K_i]$, and route the instance to the solver of the first
property it satisfies. This raises a natural scheduling question: in which
order should the guards be evaluated?

\begin{definition}[Guard and cost model]
	For each $i \in \{1,\dots,m\}$ let
	\begin{itemize}
		\item $v_i \;\triangleq\; C_{\mathrm{verify}}(K_i)$ be the {\em verification
			cost} of the guard $K_i$ (e.g.\ $O(n)$ for a linear scan, $O(|V|^9)$ for perfect graphs);
		\item $c_i$ be the worst-case running time of the fast solver $\mathrm{alg}(K_i)$ after a successful guard (locating, reading and solving the reduced problem);
		\item $p_i \;\triangleq\; \Pr_{x}\bigl[x \models K_i\bigr]$ be the
		{\em hit rate} of $K_i$, i.e.\ the probability that a random input
		satisfies $K_i$ (its coverage under the workload distribution, \S\ref{sec:Parametric}).
	\end{itemize}
	We assume the events are disjoint ($\Pr[x \models K_i \wedge x \models K_j]=0$
	for $i \neq j$); the general case is discussed in Remark~\ref{rem:overlap}.
	The dispatcher evaluates the guards in some fixed order $\pi$ and, upon the
	first index $i$ with $x \models K_{\pi(i)}$, runs
	$\mathrm{alg}(K_{\pi(i)})$ on the projected instance and returns. If no guard
	matches, the input is handled by the generic fallback (or rejected by the SLA
	gate of Section~\ref{sec:SE}).
\end{definition}

\begin{proposition}[Optimal guard ordering]
	\label{prop:order}
	Let $C_{\pi}(x)$ be the total cost incurred by the dispatcher on input $x$
	under the ordering $\pi$, including verification and the executed solver.
	Then the expected cost $\mathbb{E}_x[C_\pi(x)]$ is
	\begin{equation}
	\label{eq:expected}
	\mathbb{E}[C_\pi]
	\;=\;
	\underbrace{\sum_{i=1}^{m} p_i\, c_i}_{\text{\normalsize order-invariant}}
	\;+\;
	\sum_{j=1}^{m} v_{\pi(j)}\,
	\Bigl(1 - \textstyle\sum_{\ell=1}^{j-1} p_{\pi(\ell)}\Bigr),
	\end{equation}
	and an optimal ordering is obtained by sorting the guards in decreasing order
	of the ratio $p_i / v_i$:
	\begin{equation}
	\label{eq:smith}
	\frac{p_{\pi(1)}}{v_{\pi(1)}}\; \ge\;
	\frac{p_{\pi(2)}}{v_{\pi(2)}}\; \ge\;
	\cdots \;\ge\;
	\frac{p_{\pi(m)}}{v_{\pi(m)}}.
	\end{equation}
\end{proposition}

\begin{proof}
	Guard $\pi(j)$ is evaluated exactly when no earlier guard matched, an
	event of probability $1-\sum_{\ell<j}p_{\pi(\ell)}$, hence the indicators
	$\mathbf{1}[\pi(j) \text{ is reached}]$ have expectation
	$1-\sum_{\ell<j}p_{\pi(\ell)}$; the expected total verification cost is the
	second sum in \eqref{eq:expected}. Each solver $\mathrm{alg}(K_i)$ runs
	exactly when $x \models K_i$, with probability $p_i$ independent of the order,
	which gives the order-invariant first term. To see that
	\eqref{eq:smith} is optimal, suppose two adjacent guards $(a,b)$ are
	checked in the order $a$ then $b$ rather than $b$ then $a$.
	The difference in the corresponding expected verification costs is
	\[
	v_a + (1-p_a)v_b - \bigl[\,v_b + (1-p_b) v_a\bigr]
	\;=\; p_b v_a - p_a v_b.
	\]
	Hence checking $a$ before $b$ is no worse than the opposite iff
	$p_a v_b \ge p_b v_a$, i.e.\ $p_a/v_a \ge p_b/v_b$. Transposing adjacent
	inversions of \eqref{eq:smith} never increases the expected cost, so any
	guardwise-sorted order is optimal. 
\end{proof}

The ratio rule is classical: it is Smith's rule for minimizing the weighted sum of completion times \cite{smith1956various}, and the same exchange argument yields optimal orders in sequential testing \cite{simon1975optimal}. We restate it because it applies verbatim to structural guards.

\begin{remark}[Special cases]
	\label{rem:special}
	The rule in \eqref{eq:smith} subsumes two natural heuristics:
	(i) if all guards are equally likely ($p_i=p$), it reduces to checking the
	cheapest guard first, i.e.\ increasing $v_i$;
	(ii) if all guards cost the same ($v_i = v$), it reduces to checking the
	most {\em likely} guard first, i.e.\ decreasing $p_i$.
	The general rule balances the two: a guard should be checked early either if
	it is very likely to match ($p_i$ large) or if it is cheap to verify
	($v_i$ small).
\end{remark}

\begin{remark}[Link to the uncertainty-reduction bound]
	\label{rem:cost}
	The ordering is tied to Theorem~\ref{thm:uncertainty}. The solver cost
	$c_i$ is bounded by the location, reading and reduced-problem terms of \eqref{eq:urd-constructive},
	and the guard cost is
	$v_i=C_{\mathrm{verify}}(K_i)$. Thus the ordering
	\eqref{eq:smith} minimizes the expected total
	\[
	\mathbb{E}[C]\;=\;
	\mathbb{E}\Bigl[\textstyle\sum_{j\le j^{\ast}} v_{\pi(j)}+
	\underbrace{c_{\pi(j^{\ast})}}_{\text{solver}}
	\Bigr]\,,
	\qquad j^{\ast}=\min\{j:x\models K_{\pi(j)}\},
	\]
	and can be computed either offline from domain knowledge or online by
	bootstrapping estimates $\hat p_i$ from observed match frequencies.
	This is the precise sense in which ``checking cheaper guards first'' is
	optimal: it removes the most uncertainty about which fast solver applies per
	unit of verification cost.
\end{remark}

\begin{remark}[Overlapping properties and unknown rates]
	\label{rem:overlap}
	When the $K_i$ are not disjoint (e.g.\ an all-equal array is also sorted), the order also affects which solver runs, and the expected verification cost becomes $\sum_{j} v_{\pi(j)}\Pr\bigl[x\not\models K_{\pi(1)},\dots,x\not\models K_{\pi(j-1)}\bigr]$. Minimizing it is an instance of the (weighted) \emph{min-sum set cover}, or \emph{pipelined set cover}, problem, which is NP-hard; the greedy rule that repeatedly selects the guard with the largest conditional probability of a first match per unit cost is a $4$-approximation \cite{feige2004minsum,munagala2005pipelined}, and it reduces to the ratio rule \eqref{eq:smith} for disjoint guards. When one guard implies another ($K_i\Rightarrow K_j$) and its solver is faster, checking $K_i$ first speeds up its matches but adds its cost to the other inputs; the trade-off is captured by the same expression. When the $p_i$ are
	unknown, they can be estimated online by
	$\hat p_i(t) = M_i(t)/t$, where $M_i(t)$ is the number of matches of $K_i$
	among the first $t$ requests; the ratio order is then maintained dynamically
	in a priority queue keyed by $\hat p_i/v_i$.
\end{remark}

\begin{remark}[Comparison with classifier cascades]
	\label{rem:cascade}
	The problem is structurally similar to ordering the stages of a classifier cascade or of a filter pipeline, where the same cost/probability trade-off appears \cite{munagala2005pipelined}; the principle holds independently of the semantics of the $K_i$'s.
\end{remark}

\begin{remark}[Guarantees under unknown verification cost]
	\label{rem:uncertain-v}
	The analysis assumes the $v_i$'s are known exactly. Suppose instead that only estimates $\hat v_i$ with $v_i\le\hat v_i\le\rho\,v_i$ are available, for some $\rho\ge1$. Ordering by $p_i/\hat v_i$ then yields an expected verification cost at most $\rho$ times the optimum. Indeed, write $\mathrm{cost}_v(\pi)$ for the expected verification cost of order $\pi$ under costs $v$; it is linear and monotone in $v$. If $\hat\pi$ is optimal for $\hat v$ and $\pi^\ast$ for $v$, then $\mathrm{cost}_v(\hat\pi)\le\mathrm{cost}_{\hat v}(\hat\pi)\le\mathrm{cost}_{\hat v}(\pi^\ast)\le\rho\,\mathrm{cost}_v(\pi^\ast)$.
\end{remark}

\noindent
{\bf On the value of the proposed conditional dispatch.}
A reader might object that the software framework reduces to an ``if-then-else'': if $x\models K_i$ run $\mathrm{alg}(K_i)$. At the level of control flow it is indeed a conditional dispatcher, and a special case of algorithm selection \cite{rice1976algorithm}. What the framework adds is a \emph{contract} attached to each branch: the guard has an explicit verification bound $C_{\mathrm{verify}}(K_i)$, the branch solver has the explicit bound of Theorem~\ref{thm:uncertainty}, and the order of the guards is optimal for disjoint guards (ratio rule) and approximately optimal otherwise (Remark~\ref{rem:overlap}). An unconstrained if-then-else exposes none of these quantities. A summary is given in Table~\ref{tab:ifelse}.

\begin{table}[htbp]
	\centering
	\caption{Dispatch on structural predicates vs.\ an unconstrained conditional.
		The framework adds an explicit {\em complexity contract} to every branch of an
		if-then-else.}
	\label{tab:ifelse}
	\small
	\begin{tabular}{@{}p{3.1cm} p{5.2cm} p{5.2cm}@{}}
		\toprule
		\textbf{Dimension}
		& \textbf{Plain if-then-else} (``if $x$ satisfies $K$ run \textsc{Alg}$_x$'')
		& \textbf{$\UP[K]$ parametric dispatch} \\
		\midrule
		Guard test cost
		& unspecified; may be as costly as the problem itself
		& bounded by the verification cost
		$C_{\mathrm{verify}}(K)$, separated and made explicit \\
		\addlinespace
		Branch (solver) cost
		& not stated
		& bounded by Theorem~\ref{thm:uncertainty}:
		$C_{\mathrm{find}}+O(\Psi)+T_{\mathcal{B}}$; hence $P\in\UP[K]$ \\
		\addlinespace
		Ordering of guards
		& ad hoc; depends on how the programmer writes the {\tt if}s
		& ratio rule $p_i/v_i$, optimal for disjoint guards
		(Section~\ref{sec:GuardOrdering}); greedy $4$-approximation otherwise \\
		\addlinespace
		Composition / conflicts
		& overlapping conditions resolved arbitrarily
		& overlaps handled by the min-sum set cover formulation (Remark~\ref{rem:overlap}) \\
		\addlinespace
		Complexity certificate
		& none exposed to the caller
		& a {\em contract}: guard bound and solver bound
		are known to the API and to the SLA layer \\
		\bottomrule
	\end{tabular}
\end{table}

Finally, if the SLA cannot be met even after optimal guard ordering,  because
no structural property $K$ applies to the instance, or because the best match
still violates the latency bound, the dispatcher can of course resort to an
approximation algorithm as the ultimate fallback: on inputs $x \not\models K$
it could compute an $\alpha$-approximate solution in polynomial time, thereby
trading exactness for a guaranteed latency bound.

\subsection{\texorpdfstring{Application in Agentic AI and Planning}{Application in Agentic AI and Planning}}

The rapid emergence of \emph{Agentic AI}, where autonomous agents based on large language models (LLMs) reason, plan, invoke tools and coordinate, raises questions that the framework can help phrase precisely: when can a sub-task be delegated to an exact solver with a guaranteed running time, and when must it be handled heuristically? We sketch three uses. They are design principles rather than theorems, and their value depends on how faithfully a task can be mapped to a formal instance.

\subsubsection{Planning Under Structural Restrictions}

Plan existence for propositional STRIPS planning is PSPACE-complete \cite{bylander1994computational}, and in general an agent's planner must rely on heuristic search. A large body of work studies which restrictions on the \emph{causal graph} of a planning domain make planning tractable, and the answers are subtle. Acyclicity of the causal graph alone is \emph{not} sufficient. Plan existence is NP-complete for binary variables and polytree causal graphs, and NP-hard for multi-valued variables and chain causal graphs, while minimal plans can be exponentially long even for acyclic causal graphs \cite{gimenez2008complexity}. Plan existence is NP-complete for singly connected causal graphs, while plan generation is polynomial for binary variables and polytree causal graphs of bounded indegree \cite{brafman2003structure}. More generally, planning is tractable when the weakly connected components of the causal graph have bounded size, and intractable otherwise under a standard parameterized-complexity assumption \cite{chen2010causal}.

In the vocabulary of this paper each tractable case is a property $K$ with $\text{Planning}\in\UP[K]$, for example $K_{\mathrm{poly},d}$: ``binary variables, unary operators, and a polytree causal graph of indegree at most $d$''. Such properties can be checked in time polynomial (typically linear) in the size of the domain description. An agent can therefore run the guard before planning and, if it succeeds, invoke a planner whose running time is certified by the contract; otherwise it falls back to heuristic search.

\subsubsection{Cost-Aware Tool Selection}

A central challenge in agentic architectures is deciding when to invoke an external tool (such as an API, a database query, or a symbolic solver) and which one. A natural heuristic, in the spirit of the guard-ordering rule, scores each tool $T\in\mathcal{T}$ by the expected reduction of residual uncertainty per unit cost,
\[
\mathcal{U}_T=\frac{\Psi(n,K_s)-\Psi(n,K_{s'})}{\text{Cost}(T)},
\]
where $K_s$ and $K_{s'}$ describe what is known about the task instance before and after invoking $T$, and selects $T^\ast=\arg\max_{T}\mathcal{U}_T$. This is a value-of-information ratio \cite{howard1966information} with $\Psi$ as the uncertainty measure. Greedy selection by such a ratio is a heuristic: it is near-optimal under adaptive submodularity \cite{golovin2011adaptive}, but not in general. Applying $\Psi$ in this way also requires that the agent's knowledge state be represented as a property of a concrete input encoding.

\subsubsection{Routing Between Learned and Symbolic Solvers}

Modern AI systems combine fast learned generation (sometimes called ``System 1'') with slower deliberate or symbolic computation (``System 2''). The framework suggests a simple routing rule. Let a request be formalized as an instance $x$ of size $n$ of a problem $P$. An orchestrator runs polynomial-time guards for properties $K$ with $P\in\UP[K]$:
\begin{enumerate}
	\item If some guard succeeds, the request is routed to the corresponding exact symbolic solver, which returns a correct answer within the certified time bound, without relying on the LLM for the computation itself.
	\item Otherwise the request is handled by the LLM (or by a heuristic or approximation algorithm), without such guarantees.
\end{enumerate}
The guarantee in the first case is relative to the formalization step: it presupposes that $x$ faithfully represents the user's request, which in agentic systems is often the weakest link.

\subsection{\texorpdfstring{Quantum Computing: Promises and Speedups}{Quantum Computing: Promises and Speedups}}

The framework also offers a vocabulary for a well-known fact of quantum computing: large quantum speedups require structure. We first extend the parametric classes to quantum computation.

\begin{definition} (Uncertainty-Parametric BQP - $\UBQP[K]$)
	For a property $K$ decidable in classical polynomial time, a language $L \subseteq \Sigma^*$ belongs to the class $\UBQP[K]$ if there exist a polynomial-time uniform family of quantum circuits $M_Q$ and a polynomial $p$ such that, for every string $x$ satisfying $K$, $M_Q$ decides whether $x \in L$ with error probability at most $1/3$ in time $p(|x|)$ \cite{bernstein1997quantum,nielsen2000qic}.
\end{definition}

As in the classical case, $\UBQP[\epsilon]=BQP$, and $L\in\UBQP[K]$ iff $L\cap K\in BQP$.

\subsubsection{Unstructured Search and Grover's Bound}

Consider search in an unstructured database of $N=2^n$ items accessed through queries. Without knowledge ($K=\epsilon$) every item may be the target, so $\Psi(N,\epsilon)=N$. Classically $\Theta(N)$ queries are necessary. Bennett, Bernstein, Brassard and Vazirani \cite{bennett1997} proved that every quantum algorithm needs $\Omega(\sqrt{N})$ queries, and Grover's algorithm \cite{grover1996fast} attains this bound. The speedup is quadratic, and the cost remains exponential in $n=\log N$ when the database is given implicitly (e.g.\ by a circuit); this is one reason why Grover search is slow in practice on large data \cite{theodorakis2025encoding}. The same lower bound is the main oracle evidence for the widely believed conjecture $NP\not\subseteq BQP$.

\subsubsection{Promises and Exponential Speedups}
\label{sec:QuantumExponential}

In the query model the role of properties $K$ is precise. For \emph{total} Boolean functions, i.e.\ when $K=\epsilon$, deterministic and quantum query complexities are polynomially related, $D(f)=O(Q(f)^6)$ \cite{beals2001quantum}. Super-polynomial quantum speedups in query complexity therefore exist \emph{only} for partial functions, i.e.\ under a promise $K$ on the input. Simon's problem is the prototypical example: under the promise that $f:\{0,1\}^n\to\{0,1\}^n$ satisfies $f(x)=f(y)$ iff $y\in\{x,x\oplus s\}$ for a hidden $s$, quantum algorithms need $O(n)$ queries, whereas classical ones need exponentially many. In our vocabulary, exponential query speedups live in $\UBQP[K]$ for nontrivial promises $K$, never at $K=\epsilon$.

Shor's algorithm \cite{shor1999polynomial} illustrates a different, problem-level kind of structure. Factoring\footnote{Which is a hard problem, but not known to be NP-complete.} reduces to finding the period of $f(x)=a^x \bmod N$, and \emph{every} such function is periodic. Periodicity is therefore not a restriction on the inputs $N$, but an algebraic property of the problem, exploited through the quantum Fourier transform \cite{NielsenChuang2010} and, more generally, through the abelian Hidden Subgroup Problem \cite{Lomont2004,ChildsVanDam2010}. Accordingly, factoring is in $\UBQP[\epsilon]=BQP$, not merely in $\UBQP[K]$ for a restrictive $K$. The HSP framework does not capture all known exponential quantum speedups: quantum simulation \cite{Lloyd1996} and quantum linear-systems algorithms \cite{HHL2009} achieve advantages through different mechanisms. The emerging picture is that exponential quantum advantage requires structure, either a promise on the inputs (Simon) or algebraic structure of the problem itself (Shor), whereas the classical collapses studied in this paper exploit structure of individual instances.

\subsubsection{Limits of Quantum Advantage on Structured Inputs}
\label{sec:QuantumStructured}

The converse observation concerns heuristics such as the quantum approximate optimization algorithm (QAOA) \cite{farhi2014qaoa}, whose canonical benchmark is MaxCut. Under $K_{\mathrm{bipartite}}$ every edge crosses the bipartition, so the maximum cut equals $|E|$ and is obtained in $O(|V|+|E|)$ time by the guard itself; under $K_{\mathrm{planar}}$ MaxCut is solvable in polynomial time \cite{hadlock1975maxcut}, and planarity can be tested in linear time \cite{hopcroft1974planarity}. In both cases no quantum algorithm is needed, and no exponential quantum advantage is possible, since the problem is already classically polynomial on these classes. (Here the residual uncertainty $\Psi$ is not small; the collapse is structural.) Quantum heuristics for MaxCut are therefore of interest only on instances outside the known tractable classes, which suggests evaluating them on instances that fail all known guards, rather than on structured benchmarks that a classical guard would solve exactly.

\section{Conclusion}
\label{sec:CR}

This paper introduced a framework for stating complexity bounds that depend explicitly on checkable knowledge about the input. Costs are measured for fixed algorithms, which avoids the degeneracy of per-instance infima; the restricted worst-case cost $W(n,K)$, its promise and guarded versions, and the residual uncertainty $\Psi(n,K)$ make the price of exploiting a property $K$ explicit. The Uncertainty Reduction Theorem decomposes the guarded cost into the costs of checking $K$, of locating and reading the residual information, and of solving an explicitly defined reduced problem; when the reduced problem maps back to a small instance of the original problem, this reduction plays the role of a kernelization.

On the structural side, we identified the classes $\UP[K]$ and $\UNP[K]$ with restrictions to polynomial-time promises, and introduced coverage as the quantity that makes membership statements informative. We showed that small residual uncertainty forces sparsity (and hence, by Mahaney's theorem, non-hardness unless $P=NP$) without implying tractability (unless $E=NE$); characterized FPT and XP as uniform membership in graded families of properties; and established a free-variable threshold for SAT: polynomial for $O(\log n)$ free variables, not polynomial for $n^{\gamma}$ free variables unless $P=NP$, and requiring $2^{\Omega(n^{\gamma})}$ time under ETH. We also showed that knowledge resistance must be formulated through coverage: many checkable properties make factoring easy, but under the factoring assumption none of them has non-negligible mass on balanced semiprimes.

We illustrated three distinct mechanisms of tractability (reduction of residual uncertainty, FPT-style interface parameterization, and structure that fixes the optimal solution) on Partition, Graph Isomorphism, SAT with backdoors, Vertex Cover of bounded treewidth, and TSP on Supnick, Kalmanson and Demidenko matrices. On the practical side, the same accounting yields guarded dispatch with explicit latency contracts, an optimal guard order for disjoint guards (Smith's ratio rule) with a constant-factor approximation for overlapping ones, guarded routing for agentic systems, and a reading of quantum speedups in terms of promises.

As future work, it is worth building a software library and catalogue of $\mathrm{NP}$-hard problems and their structural collapses. Each entry would record, for a given problem $P$: a family of structural properties $\{K_i\}$ for which $P \in \UP[K_i]$, their coverage on representative workloads, the verification cost $C_{\mathrm{verify}}(K_i)$, the find cost $C_{\mathrm{find}}(n,K_i)$, the reduced-problem cost of Theorem~\ref{thm:uncertainty}, and the guard order. Given an $\mathrm{NP}$-hard task and a known instance structure, an engineer (or an AI agent) could then query the certified collapse and the price of realizing it, rather than resorting to worst-case analysis. Starting from the problems treated in this paper (Partition, SAT, Vertex Cover, and TSP), the catalogue can grow incrementally as further structural properties are identified, together with natural families $\mathcal{K}$ for which $\Psi(n,K)$ and $\mathrm{USI}(P,\mathcal{K})$ are computable or effectively approximable.

A second direction is a hardness hierarchy for the framework: which properties $K$ are hard to \emph{locate} even when the residual uncertainty $\Psi(n,K)$ is small? A $W$-analogue classifying the find cost $C_{\mathrm{find}}$ would determine when the collapse promised by the Uncertainty Reduction Theorem is genuinely exploitable. Third, for instances that satisfy no tractable property, approximation algorithms are the natural fallback, and the uncertainty metrics may guide the choice of the approximation ratio. Finally, in the quantum setting it would be interesting to identify natural promises $K$ on NP-hard problems under which quantum algorithms provably help, or under which $\UNP[K]\not\subseteq\UBQP[K]$ can be shown relative to oracles or standard conjectures.




\subsection*{Acknowledgments}
This paper assembles personal ideas and work of many summers.
%
During this long period  I enjoyed related discussions with Nicolas Spyratos.
Many thanks to Michalis Kolountzakis for spotting problems that enabled me to correct the theory,
and to Charalampos E. Tsourakakis for suggesting me to establish the relationship to average complexity,
and other feedback.
Many thanks also to Anastasia Analyti,  Dimitris Plexousakis and  Grigorios Tsagkatakis for providing me with valuable feedback.
Finally I would like to thank Marios Adamidis for giving me access to the FORTHought \cite{Adamidis2026setn}
that helped me assemble my notes and ideas
in a coherent text. 
This revised version was prepared with the assistance of the AI model Claude (Anthropic), which was used to review the previous version and to draft and edit parts of the text, definitions and proofs.

\bibliographystyle{plain}
\bibliography{paper,newrefs} 

\appendix

\section{Implementation Results}
\label{sec:ImplResults}

The code is publicly available at 
\url{https://github.com/YannisTzitzikas/PNPu}.

\subsection{SAT}

\begin{Verbatim}[frame=lines,breaklines=true,fontsize=\scriptsize]
================================================================
DESIGN-BY-CONTRACT & PARAMETRIC COMPLEXITY (P_u[K]) DEMONSTRATION
-- SAT problem --
================================================================

--> Test 1: Processing 2-CNF with N = 2,000,000 clauses (K_2CNF)...
- SAT Result        : SATISFIABLE
- #Variables        : 5,000
- #Clauses          : 2,000,000
- Matched Property K: K_2CNF (syntactic 2-SAT)
- Guard Verify Time : 34.210 ms (C_verify(K) linear scan)
- Solver Exec Time  : 1787.842400 ms (P_u[K] polynomial collapse)
- Total Time        : 1822.055 ms
- SLA Status        : PASSED [SLA MET]

--> Test 2: Processing Horn formula with N = 2,000,000 clauses (K_HORN)...
- SAT Result        : SATISFIABLE
- #Variables        : 5,000
- #Clauses          : 2,000,000
- Matched Property K: K_HORN (Horn-SAT)
- Guard Verify Time : 21.788 ms (C_verify(K) linear scan)
- Solver Exec Time  : 163.461400 ms (P_u[K] polynomial collapse)
- Total Time        : 185.250 ms
- SLA Status        : PASSED [SLA MET]

--> Test 3: Testing small unstructured 3-SAT (K = \epsilon)...
- SAT Result        : SATISFIABLE
- #Variables        : 3
- #Clauses          : 3
- Matched Property K: K_UNSTRUCTURED (Fallback DPLL)
- Guard Verify Time : 0.029 ms (C_verify(K) linear scan)
- Solver Exec Time  : 0.033200 ms (P_u[K] polynomial collapse)
- Total Time        : 0.064 ms
- SLA Status        : PASSED [SLA MET]

--> Test 4: Attempting large unstructured 3-SAT (N = 500, K = \epsilon)...
[CONTRACT REJECTION TRIGGERED]
Exception Message: Design-by-Contract Violation: Formula lacks structural property K (K=\epsilon) with 600 variables and 500 clauses exceeds deterministic SLA execution bounds.
SLA Status       : PROTECTED (System prevented unbounded CPU execution!)

================================================================
SUMMARY OF PRACTICAL VALUE:
1. Verifying property K in O(m) collapses SAT solving to O(m) in P_u[K].
2. Two million clauses (2-CNF and Horn) solved in well under the 2500 ms SLA.
3. Unstructured 3-SAT with small input passes.
4. Unstructured 3-SAT inputs are rejected upfront, guaranteeing deterministic latency.
================================================================

\end{Verbatim}

\subsection{Vertex Cover}

\begin{Verbatim}[frame=lines,breaklines=true,fontsize=\scriptsize]
================================================================
DESIGN-BY-CONTRACT & PARAMETRIC COMPLEXITY (P_u[K]) DEMONSTRATION
-- Vertex Cover problem --
================================================================

--> Test 1: Processing a TREE with V = 500,000 vertices (K_TREE)...
- Min Vertex Cover   : 250000
- #Vertices          : 500,000
- #Edges             : 499,999
- Matched Property K : K_TREE (linear DP)
- Guard Verify Time  : 59.233 ms (C_verify(K) linear scan)
- Solver Exec Time   : 23.266800 ms (UP[K] collapse)
- Total Time         : 82.503 ms
- SLA Status         : PASSED [SLA MET]

--> Test 2: Processing a BIPARTITE graph with V = 40,000 vertices (even cycle, K_BIPARTITE)...
- Min Vertex Cover   : 20000
- #Vertices          : 40,000
- #Edges             : 40,000
- Matched Property K : K_BIPARTITE (Konig)
- Guard Verify Time  : 4.727 ms (C_verify(K) linear scan)
- Solver Exec Time   : 7.038800 ms (UP[K] collapse)
- Total Time         : 11.766 ms
- SLA Status         : PASSED [SLA MET]

--> Test 3: Testing small unstructured non-bipartite graph (K = \epsilon, triangle)...
- Min Vertex Cover   : 2
- #Vertices          : 3
- #Edges             : 3
- Matched Property K : K_UNSTRUCTURED (brute force)
- Guard Verify Time  : 0.035 ms (C_verify(K) linear scan)
- Solver Exec Time   : 0.076900 ms (UP[K] collapse)
- Total Time         : 0.112 ms
- SLA Status         : PASSED [SLA MET]

--> Test 4: Attempting large non-bipartite graph (V = 200, K = \epsilon)...
[CONTRACT REJECTION TRIGGERED]
Exception Message: Design-by-Contract Violation: Non-bipartite graph with V=200 lacks structural property K (K=\epsilon) and exceeds deterministic SLA execution bounds.
SLA Status       : PROTECTED (System prevented unbounded CPU execution!)

================================================================
SUMMARY OF PRACTICAL VALUE:
1. Verifying property K in O(V+E) collapses NP-hard Vertex Cover to polynomial.
2. A 500,000-vertex tree and a 40,000-vertex bipartite graph are solved exactly,
whereas brute force would need 2^V time.
3. Testing small unstructured non-bipartite graph.
4. Unstructured non-bipartite inputs are rejected upfront, guaranteeing deterministic latency.
================================================================

\end{Verbatim}

\subsection{TSP}

Notice that here the verification cost (in Test~1) is higher than the solver cost; checking the Monge property of an $n\times n$ matrix takes $O(n^2)$ time, which is linear in the input size. The test matrix $d[i][j]=(i-j)^2$ is symmetric Monge (a Supnick matrix), so the solver outputs Supnick's tour, which is optimal (\S\ref{sec:TSPwithMonge}); for $n=200$ it has cost $794$, whereas the identity tour would cost $39800$. The identity tour is optimal for Kalmanson matrices, but not in general for Monge, Supnick or Demidenko matrices.

\begin{Verbatim}[frame=lines,breaklines=true,fontsize=\scriptsize]
================================================================
 DESIGN-BY-CONTRACT & PARAMETRIC COMPLEXITY (P_u[K]) DEMONSTRATION
                                        -- TSP problem --
================================================================

--> Test 1: Processing symmetric MONGE (Supnick) matrix, n = 200 (K_MONGE)...
  - TSP Cost          : 794
  - #Cities           : 200
  - Matched Property K: K_SUPNICK (symmetric Monge, Supnick tour optimal)
  - Guard Verify Time : 2.869 ms (C_verify(K))
  - Solver Exec Time  : 0.012567 ms (UP[K] collapse)
  - Total Time        : 2.882 ms
  - SLA Status        : PASSED [SLA MET]

--> Test 2: Processing small random matrix, n = 10 (K_UNSTRUCTURED, Held-Karp)...
  - TSP Cost          : 1943
  - #Cities           : 10
  - Matched Property K: K_UNSTRUCTURED (held-karp)
  - Guard Verify Time : 0.005 ms (C_verify(K))
  - Solver Exec Time  : 1.349159 ms (UP[K] collapse)
  - Total Time        : 1.354 ms
  - SLA Status        : PASSED [SLA MET]

--> Test 3: Attempting large random matrix, n = 30 (K = \epsilon)...
  [CONTRACT REJECTION TRIGGERED]
  Exception Message: Design-by-Contract Violation: TSP instance (n=30) lacks a structural property K (Monge/rank-1) and exceeds deterministic SLA execution bounds.
  SLA Status       : PROTECTED (System prevented unbounded CPU execution!)

================================================================
 SUMMARY OF PRACTICAL VALUE:
 1. A symmetric Monge (Supnick) matrix collapses TSP to O(n) (Supnick tour).
 2. Unstructured instances small, meet the SLA.
 3. Unstructured instances rejected upfront to guarantee deterministic latency.
================================================================
\end{Verbatim}

\end{document}